\documentclass[11pt]{article}
\usepackage{hyperref}
\usepackage{amsmath}
\usepackage{graphics}
\usepackage{color}
\usepackage{epsfig}
\usepackage{graphicx}%
\usepackage{amsfonts}%
\usepackage{amssymb}
\usepackage{setspace}
\usepackage[margin=1in]{geometry}
\usepackage{comment}
\usepackage{listings}
\usepackage{color}
\usepackage{url}

\newtheorem{theorem}{Theorem}[section]
\newtheorem{definition}{Definition}[section]
\newtheorem{claim}{Claim}[section]
\newtheorem{lemma}{Lemma}[section]

\newtheorem{corollary}{Corollary}[section]

\newtheorem{proposition}{Proposition}[section]
\newtheorem{example}{Example}[section]
\newcommand{\qed}{\hfill \mbox{\raggedright \rule{2mm}{3mm}}}
\newenvironment{proof}{\noindent{\bf Proof.}}{\qed}

\newcommand{\lbfl}{{\sc Lbfl}}
\newcommand{\cfl}{{\sc Cfl}}
\newcommand{\ufl}{{\sc Ufl}}
\newcommand{\opn}{\operatorname}

\date{}

\allowdisplaybreaks
\begin{document}

\title{
Sherali-Adams gaps, flow-cover inequalities and generalized configurations for
 capacity-constrained Facility Location 
%%\titlenote{(Produces the permission block, and
%%copyright information). For use with
%%SIG-ALTERNATE.CLS. Supported by ACM.}
}

\author{Stavros G. Kolliopoulos\thanks{Department of Informatics and
Telecommunications, National and Kapodistrian 
University of Athens, Panepistimiopolis Ilissia, Athens
157 84, Greece; (\texttt{sgk@di.uoa.gr}).}   
\and Yannis Moysoglou\thanks{ 
Department of Informatics and
Telecommunications, National and Kapodistrian 
University of Athens, Panepistimiopolis Ilissia, Athens
157 84, Greece; (\texttt{gmoys@di.uoa.gr}). } }
%\institute{Department of Informatics and
%Telecommunications, National and Kapodistrian 
%University of Athens, Panepistimiopolis Ilissia, Athens
%157 84, Greece; (\texttt{sgk@di.uoa.gr}  and \texttt{gmoys@di.uoa.gr}).
% \and
%% Department of Informatics and
% Telecommunications, National and Kapodistrian  University of Athens, Panepistimiopolis Ilissia, Athens
% 157 84, Greece; (\texttt{gmoys@di.uoa.gr}).
%}

\maketitle

\begin{abstract}
Metric facility  location is a  well-studied problem for  which linear
programming  methods have  been used  with great  success  in deriving
approximation  algorithms.  The capacity-constrained  generalizations,
such  as  capacitated  facility  location (\cfl\/)  and  lower-bounded
facility location  (\lbfl), have proved  notorious as far  as LP-based
approximation  is  concerned:   while  there  are  local-search-based
constant-factor  approximations, there  is no known  linear relaxation
with  constant  integrality gap.  According  to    Williamson and Shmoys 
devising a  relaxation-based approximation for \cfl\ is  among the top
10 open problems in approximation algorithms.

This paper  advances significantly the  state-of-the-art 
on  the effectiveness of
linear programming for capacity-constrained facility location
through  a host of impossibility 
results   
% for  substantial  families   of  linear relaxations  
for  both \cfl\  and  \lbfl.   We show  that  the  relaxations
obtained   from  the  natural   LP  at   $\Omega(n)$  levels   of  the
Sherali-Adams  hierarchy  have an  unbounded  gap,  partially 
answering an  open
question of  \cite{LiS13, AnBS13}. Here, $n$ denotes the  number of facilities
in the instance.  Building on the ideas for this result, we prove that
the standard \cfl\ relaxation enriched with the generalized flow-cover
valid inequalities  \cite{AardalPW95} has  also an unbounded  gap.  
This disproves a long-standing conjecture of \cite{LeviSS12}. 
We
finally introduce  the family of proper  relaxations which generalizes
to  its  logical extreme  the  classic  star  relaxation and  captures
general  configuration-style  LPs.  We  characterize  the behavior  of
proper  relaxations for  \cfl\ and  \lbfl\ through  a  sharp threshold
phenomenon.

\end{abstract}

\section{Introduction}

Facility location is one of the most well-studied problems in combinatorial optimization.
In the  {\em uncapacitated } version \emph{(\ufl)} we are given a set $F$ of facilities and set $C$ of clients. We may open facility $i$ by paying its opening cost $f_i$ and we may assign client $j$ to facility $i$ by paying the connection cost $c_{ij}$. We are asked to open a subset 
$F' \subseteq F$ of the facilities and  assign each client to an open
facility. The goal is to minimize the total opening and connection
cost. 
A {\em $\rho$-approximation algorithm,} $\rho \geq 1,$ outputs in polynomial time a
feasible solution with cost  at most $\rho$ times the optimum. 
The approximability of general \ufl\ is settled by an $O(\log
|C|)$-approximation \cite{Hochbaum82} which is asymptotically best
possible,  unless {\sf P = NP}. %\cite{RazS97}.
In  {\em metric} \ufl\ 
the service costs satisfy the following variant of the triangle inequality:
$c_{ij} \leq c_{ij'} + c_{i'j'} + c_{i'j}$ for any $i, i'\in F$ and $j, j' \in C.$
This very natural special case of \ufl\ is approximable within a
constant-factor, and many improved results have been published over
the years. In those,   LP-based
methods, such as filtering, randomized rounding and the primal-dual method
 have been particularly prominent (see, e.g., 
\cite{ShmoysWbook}). After a long series of papers 
the currently best approximation ratio for 
 metric \ufl\ is $1.488$ \cite{Li11}, while the best known lower
bound is $1.463,$ unless {\sf P = NP} (\cite{GuhaK99} and Sviridenko \cite{Vygen05}).
In this paper we focus on two generalizations of metric \ufl:  the 
 {\em capacitated facility location (\cfl\/)} and
 the {\em lower-bounded facility location (\lbfl\/)}.
%To our knowledge  the $1.463$ lower bound  is the only inapproximability result known
%for these two as well. 

\cfl\/ is the generalization of  metric \ufl\ where every facility $i$
has a capacity $u_i$ that specifies the maximum number of clients that
may be assigned to $i.$ In {\em uniform} \cfl\ all facilities have the
same capacity $U.$ Finding  an approximation algorithm for \cfl\/ that
uses  a linear  programming lower  bound, or  even proving  a constant
integrality  gap for an  efficient LP  relaxation, are  notorious open
problems. Intriguingly,  the following rare phenomenon occurs. 
The natural LP relaxations  have an unbounded
integrality gap and the only known $O(1)$-approximation algorithms are
based  on local  search,  with  the currently  best  ratios being  $5$
\cite{BansalGG12} for  the non-uniform and  $3$ \cite{AggarwalLBGGJ12}
for  the uniform  case respectively.  In  the special  case where  all
facility  costs  are equal,  \cfl\  admits  an LP-based  approximation
\cite{LeviSS12}.  Comparing the LP optimum against the solution output
by an LP-based  algorithm establishes a guarantee that  is at least as
strong  as the  one established  a priori  by worst-case  analysis. In
contrast, when a  local search algorithm terminates, it  is not at all
clear what  the lower  bound is.  According  to  Williamson and Shmoys
\cite{ShmoysWbook} devising a  relaxation-based algorithm for \cfl\ is
one of the top $10$ open problems in approximation algorithms.

A lot of effort has been  devoted to 
understanding the quality of relaxations obtained by an iterative 
lift-and-project procedure. Such procedures define hierarchies of 
successively  stronger relaxations, where  valid inequalities are added at 
each level. After at most $n$ levels, where $n$ is the number of 
variables, all valid inequalities have been  added and thus the integer polytope is
expressed. Relevant methods  include those  developed  by Balas et
al. \cite{BalasCC93},  Lov\'{a}sz and Schrijver \cite{LovaszS91} (for
linear and semidefinite programs), 
Sherali and Adams \cite{SheraliA90},    Lasserre  \cite{Lasserre01}
(for semidefinite programs). 
See
\cite{laurent} for a comparative discussion.
  
 The seminal work of Arora et al.  \cite{AroraBLT06}, studied integrality
 gaps  of  families of  relaxations  for  Vertex  Cover,
 including  relaxations    in  the  Lov\'{a}sz-Schrijver  (LS)
 hierarchy.  This  paper  introduced  the  use  of  hierarchies  as  a
 restricted model  of computation  for obtaining LP-based  hardness of
 approximation  results.   Proving  that  the integrality  gap  for  a
 problem  remains  large  after  many  levels of  a  hierarchy  is  an
 unconditional  guarantee   against  the  class   of  relaxation-based
 algorithms obtainable through the specific method.  At the same time,
 if an LP  relaxation maintains a gap of $g$ after  a linear number of
 levels,  one  can  take  this  as evidence  that  polynomially-sized
 relaxations are unlikely to yield approximations better than $g$ (see
 also  \cite{SchoenebeckTT07}).  In fact,  the former  belief is  now a
 theorem   for   maximum    constraint   satisfaction   problems:   
  in terms of approximation,
 LPs of size 
 $n^k,$ are exactly as powerful 
 as $O(k)$-level Sherali-Adams relaxations \cite{ChanLRS13}. 
%% it is proved that superconstant bounds on the SA
%%  hierarchy can be translated to superpolynomial bounds on the size of
%%  extended approximate relaxations for those problems.

 \lbfl\ is in a sense the opposite problem to \cfl. 
 In an \lbfl\ instance   every facility $i$ comes with a  lower
bound $b_i$ 
which is the minimum number of clients that must be assigned
 to $i$  if we open it. In {\em uniform} \lbfl\ all the lower bounds
have the same value $B.$  \lbfl\ is even less well-understood than \cfl. 
 The first approximation algorithm for the uniform case 
  had  a performance guarantee of
$448$ \cite{Svitkina08}, which has been  improved  to $82.6 $ \cite{AhmadianS12}. 
Both use local search.
%% Interestingly, the \lbfl\ algorithms
%% from \cite{Svitkina08,AhmadianS12}  both use a \cfl\ algorithm on a
%% suitable instance as a subroutine. 

Apart from some  work of the authors \cite{KolliopoulosM13,KolliopoulosM14b}
there has been no systematic theoretical  study of the power of linear programming
for approximating \cfl.  
In \cite{KolliopoulosM13} we show an unbounded gap for \cfl\ at $\Omega(n)$ levels
of the LS and the semidefinite mixed-LS$_{+}$ hierarchies, $n$ being
the number of facilities.  In  
\cite{KolliopoulosM14b} we show that linear relaxations in the classic
variables require at least an exponential number of constraints to
achieve a bounded integrality gap. Note that it is well-known that hierarchies
may produce an exponential number of inequalities already after one
round. 
 For related problems there are  some recent interesting results. 
%An improved approximation for the $k$-median problem was given in \cite{LiS13}. According to the authors their algorithm is equivalent to  rounding 
Improved  approximations  were given  for
$k$-median   \cite{LiS13}      and   capacitated   $k$-center
\cite{CyganHK12,AnBS13},  problems closely  related to
facility location.  For both, the improvements are obtained 
by LP-based techniques
that  include preprocessing  of the  instance in  order to  defeat the
known integrality gap. For $k$-median, the authors of \cite{LiS13}
state that their 
$(1+\sqrt{3} + \epsilon)$-approximation algorithm   can
be converted  to a  rounding algorithm on an 
 $O(\frac{1}{\epsilon^2})$-level LP  in the Sherali-Adams (SA)
lift-and-project hierarchy. 
They propose exploring the  direction of using 
SA  for approximating  \cfl.  In \cite{AnBS13}
the authors raise as an important question  to understand  the
power  of lift-and-project methods  for capacitated  location
problems, including  whether they automatically capture relevant 
preprocessing steps.

\vspace*{0.3cm}
\noindent
{\bf Our results.}
We give impossibility  results on arguably the most promising 
directions for   strengthening
linear relaxations for \cfl\ and \lbfl\  and in doing so we answer  open
problems from the literature.    Our contribution is threefold.

First, we show that the LPs  obtained from the
natural relaxations for  \cfl\ and \lbfl\ at $\Omega(n)$ levels of the 
SA hierarchy have an unbounded gap on an  instance
where 
 $|F|=\Theta(n)$ and $|C|= \Theta(n^3).$ 
%% First, we prove that the gap of the natural relaxation for \cfl\ (\lbfl\/)
%% remains unbounded after strengthening at 
%% the first $\Omega (n)$ levels of the Sherali-Adams
%% (SA) hierarchy, where $n$ denotes $|F|.$ 
This result answers the questions of \cite{LiS13} and
\cite{AnBS13} stated above as far as the natural LP is concerned 
and moreover it  is asymptotically tight.
In the  instances we consider clients  have unit demands and
it is well known that in this case the integer polytope and the 
mixed-integer (where fractional client assignments are allowed) 
polytope are the same.  Since SA extends to mixed-integer
programs  as well  \cite{Cornuejols08,BalasCC93},  the mixed-integer
polytope is obtained after at most  $n$ levels. Thus at most that many
levels are needed  also by the stronger, full-integer, SA procedure we employ,
which in the lifting stage multiplies 
also with assignment variables. 
From a qualitative aspect,   we  give  the first, to our knowledge, SA
bounds for a relaxation where variables have more than one type of
semantics, namely
the facility opening and the client assignment type. 
Compare this, for example, with the Knapsack and Max~Cut LPs that
contain each one type of variable. 
The lifting of the assignment variables 
raises obstacles in the proof that we managed to overcome as discussed
in Section \ref{SA-result}.

We use the \emph{local-to-global} method which was implicit in
\cite{AroraBLT06} for 
local-constraint
relaxations   and  was   then  extended   to  the   SA   hierarchy  in
\cite{FernandezdlVKM07}. 
 %% (local-global) 
See also \cite{GeorgiouMagen} for an explicit description and
\cite{CharikarMM09} for applications to Max~Cut and other
problems.  
In this approach, the feasibility of a solution for the $t$-level SA
relaxation is established through the design of a set of  appropriate
distributions over  feasible integer solutions for each constraint
such that these global distributions agree with each other locally on relevant
events.  
\begin{comment} -------- too technical 
In this approach, we interpret a linearized
product of a set $I$ of variables, namely $x_I$, as the probability of
occurrence of the event $\bigwedge_{j \in I}x_j=1,$ with respect to
a distribution over integer solutions.  
If there is an assignment $s^l$ of
values to the linearized variables
%% in any set of variables
appearing at level $t$ of SA such that for each lifted constraint there is a
distribution over some  integer solutions and the values  of the $x_I$
variables    coincide   with    the   probability    of    the   event
$\bigwedge_{j \in I }x_j=1,$ with respect to  that distribution, then
 the projection $s$ of $s^l$ on the $(y,x)$ variables
 is  feasible  for  the  relaxation  obtained  at  level  $t$  of  the
 hierarchy. 
\end{comment} 
To prove Theorem~\ref{cfl-SA:theorem} for \cfl\ we 
devise first in Lemma~\ref{assi-sym} an intuitive method to construct an initial 
 set of distributions for a %% lifted 
constraint. 
These  initial distributions are  inadequate for 
constraints where all facilities appear as indices. 
An alteration procedure,
explained in Propositions~\ref{level-ratio}--\ref{transf:prop}, 
 produces the final set of distributions. 
Theorem~\ref{cfl-SA:theorem} extends
significantly our earlier result on the  
LS hierarchy for \cfl\ \cite{KolliopoulosM13} to the stronger SA
hierarchy. It turns out that in both cases we can start from the same 
bad instance. 
It should be noted that the methodology in the two proofs is completely
different -- in \cite{KolliopoulosM13} the result was obtained via an 
inductive construction of  protection matrices.

Our second contribution (cf. Theorem~\ref{effective-cap:theorem})
is that the \emph{effective capacity}
inequalities introduced in \cite{AardalPW95,AardalPW95er} for \cfl\ 
fail to reduce the gap of the classic relaxation to constant. 
These constraints generalize the flow-cover inequalities for
\cfl. Thus  we disprove the long-standing conjecture of \cite{LeviSS12} 
that the addition of the latter % flow-cover inequalities 
to the classic LP 
suffices for  a constant integrality gap. 
Our proof deviates from standard integrality gap constructions 
by applying the local-global method.  
The bad solution fools every inequality $\pi$  because its part that is
``visible'' to $\pi$ can be extended to a solution $s^{\pi}$ that is a
convex combination of feasible integer solutions. 
%% Our proof uses the ideas from our  SA gap result
%% and relies on simple structural properties of
%% the inequalities. 
Our  ideas can be extended 
to even more general families
such as the  {\em submodular inequalities} \cite{AardalPW95}, cf. 
Theorem~\ref{thm:submod}
in the Appendix. 
All results  in this paper make no time-complexity assumptions. To our
knowledge no efficient separation algorithm for the effective
capacity inequalities is known.

We finally  introduce    
the family of  proper relaxations
which are  configuration-like linear programs.
The so-called \emph{Configuration LP} was  used by 
Bansal and Sviridenko 
\cite{BansalS06} for the Santa Claus problem and has yielded valuable insights, mostly
for resource allocation  and scheduling problems
(e.g., \cite{Svensson12}).
 The analogue of the Configuration
LP for facility location already exists, it is the {\em star
  relaxation} (see, e.g., \cite{JainMMSV03}).
%: it is the well-known 
%{\em star relaxation,} in which every  variable corresponds to a {\em star,} i.e., a
%facility $f$ and a set
%of clients assigned to $f.$ 
%The natural star relaxation 
%for   \cfl\ and \lbfl\  is  equivalent to the standard LPs 
%so it has an unbounded integrality gap.  
We take the idea of a star  to its logical  extreme by 
introducing  classes. 
A {\em class} consists of a set with an arbitrary number of facilities and clients
together with an assignment of each client to a facility in the set. 
%The definition of a class can thus vary from simple, ``local'' 
%assignments of  clients to  a  single facility, to  
% ``global'' snapshots  of the instance that 
%express  the assignment  of
%clients to a large set of  facilities.  
A {\em proper relaxation} for an instance is defined by a collection
$\mathcal{C}$ of classes and a decision variable for every class. 
We allow great freedom in 
defining  ${\cal C}\colon$ 
the only requirement   is that the resulting
formulation is symmetric and valid. 
The {\em complexity $\alpha$} of a proper relaxation is the maximum fraction
 of the 
available facilities that are contained in a class of $\mathcal{C}.$
%Proper LPs are stronger than the standard relaxation. 
%One can construct infinite families  of instances 
%where, by  increasing the complexity in a proper relaxation, one cuts off  more
%and more fractional solutions.  
In Theorem~\ref{theorem:proper}  we 
characterize the  behavior of proper relaxations  
for \cfl\ and \lbfl\ through a threshold result: 
anything less than maximum complexity results in unboundedness of
the integrality gap, while there are
proper relaxations of maximum complexity with a gap of
$1$.

Our  results disqualify  the  so far most promising approaches 
 for an efficient LP relaxation
for  \cfl. Moreover, we advance drastically the state-of-the-art for the
little understood \lbfl. 
Whether a fundamentally new approach
may succeed for either problem remains as an open question. 

\iffalse 
Our three  results  and their proofs seem additionally to suggest 
 that a bounded-gap relaxation should  not  constrain  the number of variables
 appearing in each  inequality. 
Whether such  non-trivial inequalities,   with large support and an
efficient separation oracle, exist is an open question. 
\fi 

For lack of space, some proofs and all material on \lbfl\ are
 in the Appendix.

\section{Preliminaries}
\label{sec:prel}

Given an instance $I(F,C)$ of \cfl\ or \lbfl, we use $n, m$ to denote $|F|$ and $|C|$
respectively.
We will show our negative results for uniform, integer, capacities and lower
bounds. Each client can be thought of as representing one unit of demand.
 It  is  well-known  that in such a setting  the  splittable  and
unsplittable versions  of the problem are equivalent. 
The following 0-1  IP is the standard  valid formulation of uncapacitated 
facility location with unsplittable unit demands.

\vspace*{-0.35cm}
\[ 
\begin{array}{ccc}
 \min \{ \sum_{i \in F} f_iy_i + \sum_{i \in F}\sum_{j
  \in C} x_{ij}c_{ij}  \; \mid   &    x_{ij} \leq y_i   \;\;  \forall i \in F, \forall j \in C,  \\    
\sum_{i \in F} x_{ij} =1   \;\;  \forall j \in C,   & 
y_i, x_{ij} \in \{0,1\}   \;\; \forall i \in F, \forall j \in C   \}
\end{array} \]
\vspace*{-0.2cm}

\iffalse 
\[ \min \{ \sum\limits_{i \in F} f_iy_i + \sum\limits_{i \in F}\sum_{j
  \in C} x_{ij}c_{ij}  \mid    
x_{ij} \leq y_i   \;\;  \forall i \in F, \forall j \in C    
\sum\limits_{i \in F} x_{ij} =1   \;\;  \forall j \in C   
y_i, x_{ij} \in \{0,1\}   \;\; \forall i \in F, \forall j \in C  \} \]
\fi 

\noindent 
The linear relaxation results  from the above IP by replacing the integrality constraints
 with: 
$ 
0 \leq y_i \leq 1, \; 0 \leq x_{ij} \leq 1,$   $\forall i \in F, \forall j \in C. 
$
To obtain the standard LP relaxations for 
uniform \cfl\ (and \lbfl) with capacity $U$ (lower bound $B$)
the following constraints are added
respectively: 

\vspace*{-0.2cm}
\[ \begin{array}{ccc}
\sum_ {j} x_{ij} \leq U y_i  \;\;   \forall i \in F     &  \mbox{ and }& 
\sum_ {j} x_{ij} \geq By_i   \;\;  \forall i \in F.  
\end{array}  \]
% \vspace*{-0.2cm}
%\nopagebreak

%
We will slightly abuse terminology by 
using  the term {\em (LP-classic)} for both LPs. It will be clear from the context to
which problem, \cfl\ or \lbfl,  we refer.

We proceed to define the Sherali-Adams hierarchy \cite{SheraliA90}. 
Consider a polytope $P\subseteq \mathbb{R}^d$ defined by the linear
constraints $A{x}-{b}\leq 0,$
$0 \leq x_i \leq 1$, $i = 1,\ldots,d$. We define the polytope 
$\opn{SA}^k(P)\subseteq
\mathbb{R}^d$ as follows. For every constraint $\pi(x)\leq 0$
of $P$, for every set of variables $U\subseteq \{x_i \mid i=1,\ldots,d\}$ such that
$|U|\leq k,$ and for every $W\subseteq U$, consider the valid constraint:
$\pi(x)\prod_{x_i\in U-W}x_i \prod_{x_i\in W}(1-x_i)\leq 0$.
Linearize the system obtained this way by replacing (i) $x_i^2$ with
$x_i$ for all $i$ and (ii) $\prod_{x_i\in I}x_i$
with  $x_I$ for each set $I\subseteq  \{x_i|i=1,\ldots,d\}$. $\opn{SA}^k(P)$
is the projection of the resulting linear system onto the singleton
variables. We call $\opn{SA}^k(P)$ the polytope {\em obtained from $P$ 
at level $k$} of the SA
hierarchy. Given a cost vector $c \in \mathbb{R}^d,$ the {\em relaxation
obtained from $P$ at level $k$} of SA is $\min \{c^T x \mid x \in \opn{SA}^k(P) \}.$

\begin{comment}  %%%%%%%%%%%%%%%%%%%%%%%%%%%%%%%%%%
The second well-known LP  is the star relaxation.
A {\em star} is a set consisting of some  clients and
one facility. Let $\mathcal{S}$ be a set of stars. For a star  $s\in \mathcal{S},$ let  $x_s$ be an indicator variable denoting
whether  $s$ is picked.   The cost  $c_s$ of star $s$   is equal
to the opening cost of the corresponding facility plus the cost of
connecting the star's clients to it. 
\iffalse \[ \begin{array}{ccccc}
& \min \sum_{s\in \mathcal{S}} c_sx_{s}   &  \\
\sum_{s \ni  j } x_s = 1   \;\; \forall j \in C  & &
\sum_{s \ni  i } x_s \leq  1  \;\;  \forall i \in F  &  & 
x_s \geq 0         \;\;  \text{for all stars } s \in \mathcal{S}  
\end{array}  \]    \fi
\[ \min \{ \sum_{s} c_sx_{s} \mid \sum_{s \ni  j } x_s = 1   \;\;
\forall j \in C, \;\;\;  \sum_{s \ni  i } x_s \leq  1  \;\;  \forall i \in F,
\;\;\; x_s \geq 0         \;\;  \text{for all stars } s \in \mathcal{S}   \} \]

Defining $\mathcal{S}$ as the set of all stars $s$ where  the total number of the clients in $s$ 
is at most the
capacity $U$ (at least the bound $B$),  we get corresponding relaxations for
  \cfl\ (\lbfl).
In the rest of the paper we  slightly abuse terminology by 
using    {\em (LP-star)}  to refer to the
star relaxation for the problem we examine each time (\cfl\  or \lbfl).

It is well known  that for both \cfl\ and \lbfl, (LP-classic) and (LP-star) are
equivalent, therefore (LP-star) can be solved in  polynomial time.
%%%%%%%%%%%%%%%%%%%%%%%%%%%%%%%%%% 
\end{comment}

\section{Sherali-Adams gap for \cfl }\label{SA-result}

%%In this section we  prove  tight bounds on the number of levels of
%% SA needed to reduce the gap of (LP-classic) below
%% $\Theta(n)$.  
Consider an  instance of metric  \cfl\ with a  total of
$2n$ facilities,  $n$ with opening cost  $0$ which we  call cheap (and
denote the corresponding set by $Cheap$) and $n$ with opening cost $1$
which   we  call  costly   (and  denote   by
$Costly$).  The  capacity  $U=n^3$  and  we have  a  total  of  $nU+1$
clients. All connection costs are $0$. We will show that the following
bad solution  $s$ to the instance\footnote{The reader should notice
  that any similarity with Knapsack is
  superficial. Theorem~\ref{cfl-SA:theorem} is about the % fractional 
{\em \cfl}
    polytope. Moreover, it is easy to embed our instance 
 in a slightly larger one, with a non-trivial metric,  so that 
%% the bad
%%  \cfl\ solution and the optimal solution of the underlying
%% ``knapsack'' problem have the same cost. 
 the projection of the bad \cfl\ solution
to the $y$-variables, is in the integral polytope of the ``underlying''
     knapsack  instance.
%  
%% one larger instance has a dummy facility at distance 1 whose cost 
%% is n**y_co. In the bad CFL solution we keep him at zero. 
%% The Knapsack optimal solution opens the n cheap and the dummy. 
%
%% if you want projection, open the dummy instance at 1 and let it
     %% have cost zero. It has no clients.   Projecting the y-variables gives a
 %% ``knapsack'' solution of cost slightly above optimal that opens
 %% n+1+1/n
%% It is however in the integral polytope. 
% Both  methods have  the advantage that the SA proof carries through
% for the new instance and the new solution! 
}
  survives a  number of SA levels, which is  linear 
in the number $2n$ of facilities. On the other hand, it
is known that 
at level $2n$ the  relaxation obtained expresses the integral polytope.
Let  $\alpha=n^{-2}$.  For all  $i\in  Cheap$  and  for all  $j\in  C,$
$y_{i}=1$ and  $x_{ij}=\frac{1-\alpha}{n}$, and for  all $i\in Costly$
and     for     all     $j\in    C$     $y_{i}=\frac{10}{n^2}$     and
$x_{ij}=\frac{\alpha}{n}$. Theorem~\ref{cfl-SA:theorem} below
indicates that, as often with hierarchies,  simple valid inequalities are
generated after many rounds.  The reader who is further interested in
the robustness of SA for \cfl\ may consult Section~\ref{sec:robust} in
the Appendix.

The following lemma, which is implicit in previous work
\cite{FernandezdlVKM07,GeorgiouMagen} 
 gives sufficient conditions for a solution to be feasible at  level $k$ of the SA hierarchy.

\begin{lemma}\cite{FernandezdlVKM07,GeorgiouMagen}\label{SA-survival}
Let $s$ be  a feasible solution to the relaxation  and let $v(\pi ,z)$
be  the  set  of  variables  appearing  in  a  lifted  constraint
obtained from $\pi$
multiplied by  $z$.  Solution $s$ survives $k$ levels of
SA if  for every constraint $\pi$  and each multiplier $z$  with at most
$k$ distinct variables there is:
\vspace*{-0.2cm}
\begin{itemize}
\item[1] A solution  $s'=s_{\pi,z}$  which agrees with
  $s$  on  $v(\pi ,z)$   such that $s'$
is a convex combination $E_d$ of integer solutions  (and thus $E_d$ defines a distribution on integer solutions) and
\item[2]  For   any  two  sets  $v(\pi_1,z_1)$   and  $v(\pi_2,z_2)$,  let
  $x_1x_2\cdot\ldots\cdot x_{l},$ $l \leq k+1,$ be a product appearing
  in both lifted constraints obtained from $\pi_1$ and $\pi_2$
  multiplied 
  with $z_1$ and $z_2$ respectively.  
Then
  the probability  $P[x_1=1 \wedge x_2=1  \wedge\ldots\wedge x_{l}=1]$ is
  the  same in  both  distributions $E_{d_1}$  and  $E_{d_2}$ associated  with
  $v(\pi_1,z_1)$ and $v(\pi_2,z_2)$ respectively.
\end{itemize}
  \end{lemma}

First   consider  a  constraint   $\pi\colon  \sum_{j}x_{i^{\pi}j}\leq
Uy_{i^{\pi}}$ and a multiplier $z$.   After multiplying by $z$ and
expanding,  we obtain
a   linear  combination   of  monomials
(products). Then, % similarly to the previous, 
for the  $k< n-1$ levels
we  consider   there  must  be   some  costly  facility   $i_b  \notin
v(\pi,z)$.  We construct  a solution  $s_{\pi ,z}=(y',x')$  by setting
$y'_{i_b}=1-\sum_{i  \in  Costly-\{i_b\}}y_i$  and letting  all  other
variables the  same as in the  original bad solution $s$.  We say that
facility  $i_b$  \emph{takes the  blame}.  We  will   prove  that
$s_{\pi ,z}$ can be obtained as a convex combination $E_d$ of a set of
integer   solutions    satisfying   constraint   $\sum_{i\in   Costly}
y_i=1$.  %Here is a crucial difference to the mixed case: 
While
$s_{\pi,z}$ can be obtained as a convex combination $E_d$ in a variety of
ways, we require that the assignments of clients to the cheap facilities are
indistinguishable in $E_d$ and the same must be true for the assignments to
costly facilities other than $i_b$.  In the upcoming definition, we 
use  the product $p=z_1z_2\ldots z_l$ as an  abbreviation
of the event $\mathcal{E}_p := \bigwedge_{i=1}^{l} z_i = 1.$ 

\begin{definition}
 Let $i_b$ be the facility that takes the blame. We say that a distribution $E_d$ is \emph{assignment-symmetric}  if the following are true:
\vspace*{-0.1cm}
\begin{itemize}
\item[1] $P_{E_d}[x_{i_{a_1}j_{b_1}}\ldots x_{i_{a_t}j_{b_t}}y_{i_{a_{t+1}}}\ldots y_{i_{a_{l}}}]$,  with $t+l \leq k+1$ is the same if we exchange all occurrences of cheap facility $i_{r}$ by cheap facility $i_{r'}$ (in other words relabeling facilities). Note that we allow repetitions of facilities and clients in the description of the event.

\item[2]  $P_{E_d}[x_{i_{a_1}j_{b_1}}\ldots x_{i_{a_t}j_{b_t}}y_{i_{a_{t+1}}}\ldots y_{i_{a_{l}}}]$ is the same if we exchange all occurrences of client $j_q$ by client $j_{q'}$.

\item[3]  $P_{E_d}[x_{i_{a_1}j_{b_1}}\ldots x_{i_{a_t}j_{b_t}}y_{i_{a_{t+1}}}\ldots y_{i_{a_{l}}}]$ is the same if we exchange all occurrences of costly facility $i_1$ by costly facility $i_2$, $i_1,i_2\neq i_b$.
\end{itemize}
\end{definition}
\vspace*{-0.1cm}

We can always obtain $s_{\pi ,z}$ from such an assignment-symmetric distribution $E_d$ as shown
in the following lemma. 

\begin{lemma}\label{assi-sym}
Solution $s_{\pi ,z}$ is a convex combination $E_d$ of integer solutions which defines an assign\-ment-symmetric distribution.
\end{lemma}

\begin{proof}
We describe  a probabilistic experiment which induces an assignment-symmetric distribution $E_d$ over integer solutions satisfying $\sum_{i\in Costly} y_i=1$.

 Fix           costly           facility          $i_b$.           Let
 $w^1_{i_b}=\frac{\sum_{j}x'_{i_bj}}{y'_{i_b}}$ be  the desired number
 of  clients assigned  to facility  $i_b$ in  the integer solutions in
 $E_d$ where facility $i_b$ is  opened. 
 To simplify the presentation let us assume  that $w^1_{i_b}$ and the
 $w$ values we subsequently define are integers  
(we  discuss in the
 Appendix  how to handle fractional $w$'s).  Let
 $w^1_{i_{ch}}=\frac{|C|-w_{i_b}}{|Cheap|}$  be the number  of clients
 assigned  to facility $i_c,c\in  Cheap$. Likewise,  fix costly facility
 $i_{co}\neq                         i_b$.                         Let
 $w^2_{i_{co}}=\frac{\sum_{j}x'_{i_{co}j}}{y'_{i_{co}}}$ be the number
 of clients assigned to facility  $i_{co}$ in each integer solution in
 $E_d$   where  facility   $i_{co}$  is   opened  and   similarly  let
 $w^2_{i_{ch}}=\frac{|C|-w_{i_{co}}}{|Cheap|}$   be   the  number   of
 clients  assigned  to  facility  $i_c,c\in  Cheap,$  in  each  integer
 solution in $E_d$ where facility $i_{co}$ is opened. Observe that all
 the defined $w$'s are less than $U$. The following procedure produces
 the assignment-symmetric distribution $E_d$.

Pick costly facility $i_c$ with probability $y'_{i_c}$. If $i_c=i_b$ ($i_c\neq i_b$) then consider $n$ bins corresponding to the $n$ cheap facilities  each one having $w^1_{ch}$ ($w^2_{ch}$) slots and $1$ bin corresponding to $i_{co}$ having $w^1_{i_b}$ ($w^2_{co}$)  slots.
Randomly distribute $|C|$ balls to the slots of the $n+1$ bins, with exactly one ball in each slot.
Note that the above experiment induces a distribution over feasible integer solutions satisfying $\sum_{i\in Costly} y_i=1$ since all the defined bin capacities are less than $U$ and every client is assigned to exactly one opened facility in each outcome and exactly $1$ costly facility is opened. 
Moreover the induced distribution $E_d$ is assignment-symmetric and
the expected $(y,x)$ vector with respect to $E_d$ is solution $s_{\pi
  ,z}$. 

Clearly, $s_{\pi ,z}$  is the convex combination induced  by $E_d$ and
$E_d$ is  assignment-symmetric: the cheap facilities  are always open,
and the costly  are open a fraction  of the time that is  equal to the
value  of  their  corresponding  $y$  variable.  The  expected  demand
assigned to  each $i_{co} \in Costly$ is  $y'_{i_{co}}w_{co}$ which is
the  total  demand assigned  to  $i_{co}$  by  $s_{\pi,z}$. Since  the
clients  have  the  same  probability  of  being  tossed  in  the  bin
corresponding to $i_{co}$, the  expected assignment of each client $j$
to $i_{co}$ is the same as in $s_{\pi,z}$.  
Similarly we can prove
that the expected assignments to the cheap facilities are as required,
see the Appendix for details.  
\end{proof}

We set the product-variables $x_I$
appearing in constraint $\pi$ multiplied by multiplier $z$ to 
$P_{E_d}[I]$. Constraints $x_{ij} \leq y_i,$ $x_{ij} \leq 1,$
$y_{i} \leq 1,$ 
%\ref{eq:x<y} \ref{eq:intx} and \ref{eq:inty} 
 are handled 
in the exact same way; the set of variables appearing
in them is a
 % proper 
subset of those  % of variables 
appearing in the more complex constraints.

The  second and  more challenging  case  is when  constraint $\pi$  is
$\sum_{i}x_{ij^{\pi}}=1$ for some client $j^{\pi}$. Let again $z$ be a
multiplier of level $k$. Observe now that all facilities in $F$ appear
in $v(\pi,z)$ as indexes of at least the $x_{ij}$ variables. We select
one facility $i_b$ not appearing  in $z$ to \emph{take the blame}. Let
$s_{\pi ,z}=(y',x')$  be the corresponding extended  solution that can
be written as  a convex combination/assignment--symmetric distribution
$E_d$  of integer solutions; the  existence of  $E_d$ is  ensured by
Lemma \ref{assi-sym}.  In this case there  is a major  obstacle to the
agreement  of   the  products  $x_{I}$:  conditioning   on  the  event
$x_{i_bj}$ the  probability of an event $x_{ij'},i\in  Cheap$ for some
$j'\neq j$ is higher  than it would be if we were  to condition on the
event  $x_{i'j},i'\in  Costly-\{i_b\}$.  The  same is  true  for  more
complex events involving  assignments to cheap facilities conditioning
on an  assignment of  facility $i_b$ compared  to the  analogous event
conditioning on  some other costly  facility. This can  be problematic
since facility  $i_b$ takes the  blame in some distributions  but does
not  in some  others and  thus there  is the  danger of  violating the
consistency    required    by    the    2nd   condition    of    Lemma
\ref{SA-survival}. We  overcome this difficulty  by making alterations
to $E_d$ and constructing a distribution $E_f$ where the probabilities
of the aforementioned events are the same.

We now  devise the altered  distribution $E_f$.  We first  display the
intuition in the following example: consider the event $A \colon x_{i_bj}=1
\wedge x_{i_{ch}j'}=1$ and the event $B \colon x_{i_{co}j}=1 \wedge
x_{i_{ch}j'}=1$  with  $i_{co}\in   Costly-\{i_b\}$  and  $i_{ch}  \in
Cheap$.         The        probability        of         $A$        is
$P[A]=P[x_{i_bj}=1]P[x_{i_{ch}j'}=1 \mid x_{i_bj}=1]=x'_{i_bj}\frac{w^1_{ch}}{|C|-1}$
and         the          probability         of         $B$         is
$P[B]=P[x_{i_{co}j}=1]P[x_{i_{ch}j'}=1 \mid x_{i_{co}j}=1]=x'_{i_{co}j}\frac{w^2_{ch}}{|C|-1}$.
Note  that  $P[A]\approx  P[B](1+1/n)$  so  $P[A]$  is  only  slightly
greater.  We nullify  the  difference between  those probabilities  by
performing  an alteration  step  to distribution  $E_d$  that we  call
\emph{transfusion of probability}. We  pick some measure of an integer
solution  $s_1$  for which  $x_{i_{ch}j'}=1  \wedge x_{i_bj}=1  \wedge
x_{i_bj''}=0$  for some  client $j''$.  We pick  the same  quantity of
measure of some  integer solution (or of some  set of solutions) $s_2$
for which  $x_{i_{ch}j'}=0 \wedge x_{i_bj}=0  \wedge x_{i_bj''}=1$ and
we exchange the values of the assignments $x_{i_bj},x_{i_bj''}$ of the solutions.
  Let that quantity be $P[A]-P[B]$, it is easy to
see  that each  set of  solutions has  enough measure  to  perform the
transfusion. The resulting distribution  $E_f$ now has $P[A]=P[B]$. In
general, when transfusing probabilistic measure for complex events, we
must be  careful not  to change the  probability of  events involving
only  assignments to cheap facilities, as opposed to the simplified
example above.

Now let  $p$ be a product  appearing in constraint  $\pi$ after having
multiplied by multiplier $z$.  We only consider products where exactly
one variable $x_{i_bj}$  appears. Recall we chose   $i_b$ so that it
does not appear in $z$; thus  
we  cannot have  $y_{i_b}$ or  more  than one
assignments  of $i_b$ appearing  in a  product $p.$
  We  may also  assume  that there  is no  $y_i$
variable in $p$, since if  there is for some $i\in Costly-\{i_b\}$ the
probability of $\mathcal{E}_p$ is simply $0$ and if $i\in Cheap$ the we can ignore
the effect of $y_i=1$ since it is always true. Likewise we assume that
there is no assignment variable  of another costly facility. We shall
make corrections of the probability of all such events $\mathcal{E}_p$ in a
top-down manner: at step $i$ we  fix the probability of all the events
$x_{i_bj}=1  \wedge  x_{i_{a_1}j_{b_1}}=1       % x_{i_{a_2}j_{b_2}}=1
\wedge   \ldots \wedge x_{i_{a_{k-i+1}}j_{b_{k-i+1}}}=1$     where    $x_{i_bj}
x_{i_{a_1}j_{b_1}}
% x_{i_{a_2}j_{b_2}}
\ldots x_{i_{a_{k-i+1}}j_{b_{k-i+1}}}$ is a product $p$
appearing in constraint $\pi$ multiplied by $z$. In other words, we fix
the probabilities  in decreasing order  of the cardinality of  the set
of variables appearing in $p$.  The following proposition relates the
probability       of      $\mathcal{E}_p$      with       that      of
$\mathcal{E}_{p'}=\mathcal{E}_{px_{ij}}$, an event with the additional
requirement that  $x_{ij}=1$.

\vspace*{-0.1cm}
\begin{proposition}\label{level-ratio}
Let $p=x_{i_bj} x_{i_{a_1}j_{b_1}} x_{i_{a_2}j_{b_2}}\ldots  x_{i_{a_{l}}j_{b_{l}}}$ and let $p'=px_{i_{a_{l+1}}j_{b_{l+1}}}$. Then
in $E_d$, $(1-o(1)){P[\mathcal{E}_p]}/{n}\leq P[\mathcal{E}_{p'}] \leq (1+o(1)){P[\mathcal{E}_p]}/{n}$.
\end{proposition}
\vspace*{-0.1cm}

\begin{comment}
\begin{proof}
Since the considered distribution is assignment-symmetric, event $\mathcal{E}_p$ is equivalent  to the event 
of randomly distributing $l+1$ balls to the slots of $n+1$ bins, with at most one ball in each slot, each bin having $w^1_{ch}$ identical slots, asking that ball $j$ is tossed in the bin of $i_b$ and ball $j_{r}$ is tossed in bin $i_{r}$. Since there are $\Theta (n^3)$ slots in each bin and the balls are at most $n$, it is easy to see
that $(1-o(1))\frac{P[\mathcal{E}_p]}{n} \leq P[\mathcal{E}_{p'}] \leq (1+o(1))\frac{P[\mathcal{E}_p]}{n}$.
\end{proof}
\end{comment}

\iffalse---------- expanded version 
\begin{sloppypar}
Consider step $i$ of the above iterative construction of $E_f$. Let
$p=x_{i_bj} x_{i_{a_1}j_{b_1}}
x_{i_{a_2}j_{b_2}}\ldots x_{i_{a_{k-i+1}}j_{b_{k-i+1}}}$ and the event
$\mathcal{E}_p \colon x_{i_bj}=1 \wedge x_{i_{a_1}j_{b_1}}=1 \wedge
x_{i_{a_2}j_{b_2}}=1 \wedge \ldots \wedge
x_{i_{a_{k-i+1}}j_{b_{k-i+1}}}=1$. We wish in $E_f$ the probability
$P[\mathcal{E}_p]$ to be equal to
$P[\mathcal{E}_{p/fixed}]=P[x_{i^*j}=1  \wedge x_{i_{a_1}j_{b_1}}=1
  \wedge x_{i_{a_2}j_{b_2}}=1 \wedge \ldots \wedge
  x_{i_{a_{k-i+1}}j_{b_{k-i+1}}}=1]$ in $E_d$ for $i^*\in
Costly-\{i_b\}$. We bound the ratio $\frac{P[\mathcal{E}_p]}{P[\mathcal{E}_{p/fixed}]}\colon$
\end{sloppypar}
----------- end expanded version \fi 
\begin{sloppypar}
% \noindent
Consider step $i$ of the above iterative construction of $E_f$. Let
$p=x_{i_bj} x_{i_{a_1}j_{b_1}}\ldots x_{i_{a_{k-i+1}}j_{b_{k-i+1}}}$ and the event
$\mathcal{E}_p \colon x_{i_bj}=1 \wedge x_{i_{a_1}j_{b_1}}=1 \wedge
\ldots \wedge
x_{i_{a_{k-i+1}}j_{b_{k-i+1}}}=1$. We wish in $E_f$ the probability
$P[\mathcal{E}_p]$ to be equal to
$P[\mathcal{E}_{p/fixed}]=P[x_{i^*j}=1  \wedge x_{i_{a_1}j_{b_1}}=1
  \wedge  \ldots \wedge
  x_{i_{a_{k-i+1}}j_{b_{k-i+1}}}=1]$ in $E_d$ for $i^*\in
Costly-\{i_b\}$. We bound the ratio $\frac{P[\mathcal{E}_p]}{P[\mathcal{E}_{p/fixed}]}\colon$
\end{sloppypar}

\vspace*{-0.1cm}
\begin{proposition}\label{transf_fraction}
\begin{sloppypar}
Let $\mathcal{E}_p$ and $\mathcal{E}_{p/fixed}$ be defined as above. 
%% Let $\mathcal{E}_p:x_{i_bj}=1 \ \wedge x_{i_{a_1}j_{b_1}}=1 \wedge
%% x_{i_{a_2}j_{b_2}}=1 \wedge ...x_{i_{a_{k-i+1}}j_{b_{k-i+1}}}=1$ and
%% $\mathcal{E}_{p/fixed}:x_{i^*j}=1  \wedge x_{i_{a_1}j_{b_1}}=1 \wedge
%% x_{i_{a_2}j_{b_2}}=1 \wedge ...x_{i_{a_{k-i+1}}j_{b_{k-i+1}}}=1$ with
%% $i^*\in Costly-\{i_b\}$. 
Then  \\
%$\frac{P[\mathcal{E}_p]}{P[\mathcal{E}_{p/fixed}]}<e^2$.
$(1+(1-o(1))1/n)^{k-i+1}\leq
\frac{P[\mathcal{E}_p]}{P[\mathcal{E}_{p/fixed}]}\leq
(1+(1+o(1))1/n)^{k-i+1}$.
\end{sloppypar}
\end{proposition}
\vspace*{-0.1cm}

\begin{comment}
\begin{proof}
Consider again the random experiment of the proof of Proposition \ref{level-ratio}. Recall that
$ w^2_{ch} \leq w^1_{ch} \leq w^2_{ch}(1+1/n)$. Note that again both $w^1_{ch},w^2_{ch}$ are $\Theta (n^3)$ and $k-i+1<n$. In the ball tossing experiment, the probability of success of ball $j_r$ in the case where the capacities of the cheap bins is $w^1_{ch}$ is at most $(1+2/n)$ the probability of success of the same event in the case where the capacities of the cheap bins is $w^2_{ch}$ (we are very generous here). So the ratio $\frac{P[\mathcal{E}_p]}{P[\mathcal{E}_{p/fixed}]}\leq (1+2/n)^{k-i+1} < e^2$ using that $\lim_{x \rightarrow \infty}(1+d/x)^{x}=e^d$.
\end{proof}
\end{comment}

The corrections of the  probabilities of events of previous iterations
affect the probabilities of the events of the current iteration of the
procedure  that  constructs  $E_f$.   We  bound  this  effect  on  the
probability of  an event $\mathcal{E}_p$ of the  current iteration $i$
by     considering      the     corrections     of      the     events
$\mathcal{E}_{p'}=\mathcal{E}_p \wedge x_{ij}=1$, with $x_{ij}$ in the
set of variables appearing in $z$ and $x_{ij}\notin \mathcal{E}_p$, of
the previous iteration and using the union bound.\footnote{Notice that
  any  effect of iteration  $j< i-1$  on $P[\mathcal{E}_p]$, originates
  from events that are subsets of $\mathcal{E}_{p'}$ and has therefore
  been  accounted for.}   There are  exactly $i$  events needed  to be
taken into consideration for  each such $\mathcal{E}_p$ of the current
step $i$.  The amount of the  effect of the correction of the previous
iteration   is   by    Proposition   \ref{transf_fraction}   at   most
$i((1+(1+o(1))1/n)^{k-i+2}-1)P[\mathcal{E}_{p'/fixed}]$
%(picking the $\mathcal{E}_{p'}$ with largest $P[\mathcal{E}_{p'/fixed}]$)
 while the measure of the  needed correction for $\mathcal{E}_p$ is at
 least  $((1+(1-o(1))1/n)^{k-i+1}-1)P[\mathcal{E}_{p/fixed}]$ which by
 Proposition \ref{level-ratio} and by the number of rounds we consider
 is                higher,                in                particular
 $((1+(1-o(1))1/n)^{k-i+1}-1)P[\mathcal{E}_{p/fixed}]\geq
 n(1-o(1))((1+(1-o(1))1/n)^{k-i+1}-1)P[\mathcal{E}_{p'/fixed}]>i((1+(1+o(1))1/n)^{k-i+2}-1)P[\mathcal{E}_{p'/fixed}]$. To
 subtract  from  $P[\mathcal{E}_p]$  the  rest  of  the  probabilistic
 measure required from  the correction, say a measure  of $\mu$, we do
 the  following transfusion step:  pick a  measure $\mu$  of solutions
 from distribution $E_d$ such that $x_{i_bj}=0$, $x_{i_bj'}=1$ for any
 $j'$ such that $x_{i_bj'} \notin  v(\pi ,z)$, all the other events of
 $\mathcal{E}_p$ are false, and so are all the remaining events
 corresponding to  assignments in $z$. Likewise pick  an equal measure
 of solutions  from $E_d$  such that $x_{i_bj}=1$,  $x_{i_bj'}=0$ with
 $x_{i_bj'} \notin v(\pi ,z)$, all the other events of $\mathcal{E}_p$
 are true,  and all the remaining events  corresponding to assignments
 in  $z$  are  false.  Now  exchange the  values  of  the  assignments
 $x_{i_bj}$  and $x_{i_bj'}$  of the  solutions of  the two  sets. The
 resulting distribution  has the probability  of $\mathcal{E}_p$ fixed
 and moreover,  by the  choice of  the sets of  solutions on  which we
 perform the transfusion step, the  probability of the events fixed in
 previous iterations  was not altered and neither  was the probability
 of events containing only assignments of cheap facilities.
%(the "self-correlation" of the assignments of the cheap facilities is
%invariant)
Clearly, the solution $s_{\pi,z}$ is still obtained in expectation. 
 It remains to show that the transfusion step can be performed, i.e., that
  there is enough  measure $\mu$  in the  involved sets  of integer
 solutions.

\vspace*{-0.1cm}
\begin{proposition}\label{transf:prop}
The probabilistic transfusion step of the above iterative procedure can always be performed.
\end{proposition}
\vspace*{-0.5cm}
\begin{comment}
\begin{proof}
Consider the measure $t$ in $E_d$ of the set of integer solution satisfying  $y_{i_b}=1$ 
and all events in $\mathcal{E}_p$ being false, namely $x_{i_bj}=0
\wedge x_{i_1j_1}=0 \wedge x_{i_2j_2}=0 \wedge \ldots \wedge x_{i_{k-i+1}j_{k-i+1}}=0 $. Then, by the random experiment of the construction of $E_d$, this event is equivalent to the event that facility $i_b$ is picked, $x_{i_bj}=0$ and  the $k-i+1$ balls corresponding to the clients of the rest of events
are not tossed in their corresponding bins. Using again that both $w^1_{ch},w^2_{ch}$ are $\Theta (n^3)$ and $k-i+1<n$, we can bound the probability of the $k-i+1$ balls  by that of $k-i+1$ Bernoulli trials with probability of success $2/n$ (we are once again very generous). Then the probability that all events on ball $j_i$  fail is $> (1-2/n)^{k-i+1}> \lim_{n \rightarrow \infty }(1-2/n)^n = 1/e^2$. Thus measure $t$ is at least $(y_{i_{b}}-x_{i_bj}) 1/e^2$ which is constant. On the other hand the measure required by the transfusion step for each event $\mathcal{E}_{p}$ of iteration $i$ that needs to be fixed is at most $(e^2-1)P[\mathcal{E}_{p/fixed}]=\Theta (1/n^i)$. There are $k+1\choose{k-i+1}$ such events of iteration $i$, and summing over all the iterations of our construction we get $\sum_{i=1}^{k} {k+1\choose k-i+1} \Theta (1/n^i) $ which quantity is less than $(y_{i_{b}}-x_{i_bj}) 1/e^2$ for the $k=n/10$ levels of SA we consider, so we can always pick the required amount of measure. 
%We did not make any efforts to optimize the constants.
\end{proof}
\end{comment}

\begin{theorem}\label{cfl-SA:theorem}
There is  a family of   \cfl\  instances with $2n$ facilities and $n^4+1$ clients  
such that the  relaxations   obtained   from
(LP-classic)    at $\Omega(n)$ levels of the 
Sherali-Adams  hierarchy have an  integrality gap of $\Omega(n).$
\end{theorem}

\begin{proof}
For each lifted constraint $\pi$ multiplied by multiplier $z$ at level
$t$,  the  corresponding distribution  $E_d$  or  $E_f$  is clearly  a
distribution over  integer solutions, so the first  condition of Lemma
\ref{SA-survival} is satisfied. For the second condition, observe that
if an event $\mathcal{E}_p$ involves more than one costly facility, it
has $0$ probability in  all distributions. If an event $\mathcal{E}_p$
involves only  cheap facilities,  it has the  same probability  in all
distributions  $E_f$  and  $E_d,$  since  in  the  construction  of  a
distribution $E_f$ we took care  not to change the probability of such
events.   An  event  $\mathcal{E}_p$   that  involves  more  than  one
assignment of  a costly  facility (but no  other costly) has  in every
distribution $E_f$ the same probability (which is the same as in every
$E_d$)  since  in the  construction  of $E_f$  we  did  not alter  the
probabilities   of   such   events.   And  lastly,   when   an   event
$\mathcal{E}_p$  involves  exactly   one  assignment  of  some  costly
facility $i_x$, note that in some cases $i_x$ takes the blame but in
other  cases it  does not,  depending on  $v(\pi,z)$. But  due  to the
iterative procedure  of probabilistic transfusion,  the probability of
event  $\mathcal{E}_p$ in  a distribution  in which  $i_x$ is  not the
facility that takes the blame is  equal to the probability of the same
event  in the  distributions  that  $i_x$ takes  the  blame. So  Lemma
\ref{SA-survival} holds.  It is easy to see that bad solution has  cost  
$\Theta(n^{-1})$ while any feasible solution to the instance has cost $\Omega(1)$.
\end{proof}

\section{Fooling the effective capacity inequalities for \cfl\ }\label{flow-cover}

In this section we show that the (LP-classic) for \cfl\ with the addition
of the effective capacity inequalities proposed in \cite{AardalPW95} has unbounded gap. 

Consider the  general case where  facility $i$ has capacity  $u_i$ and
client $j$  has demand  $d_j$.  For  a set $J$  of clients,  we denote
their total demand by $d(J)=\sum_{j\in  J}d_j$. Let $J \subseteq C$ be
a set of clients, let $I \subseteq F$ be a set of facilities, and let $J_i \subseteq J$ be a set of clients for each facility $i\in I$. Given a  facility $i$, we denote the  \emph{effective capacity} of
$i$ with respect to $J_i$ by $\bar{u}_i=\min \{ u_i,d(J_i) \}$.
 $I$ is a \emph{cover} with respect to $J$  if $\sum_{i \in I}\bar{u}_i = d(J) +
\lambda$  with $\lambda  > 0$.  $\lambda$ is  called  the \emph{excess
  capacity}. Let $(x)^+ = \max\{x,0\}$.  In the case where $J_i=J$ for
all  $i\in  I$  the  following inequalities  called  \emph{flow-cover}
inequalities were introduced for \cfl\ in \cite{AardalPW95}. 

\vspace*{-0.1cm}
\begin{center}
$\sum_{i\in I}\sum_{j\in J}d_jx_{ij} +\sum_{i\in I}({u}_i-\lambda)^+(1-y_i)\leq d(J)$
\end{center}
\vspace*{-0.2cm}

\begin{comment}
For the families of instances that we consider with uniform capacities and unit client demands, the above inequalities are simplified to:
\begin{center}
$\sum_{i\in J}\sum_{j\in K}x_{ij} +\sum_{i\in I}(U-\lambda)^+(1-y_i)\leq |K|$
\end{center}
\end{comment}

If $\max_{i\in I} (\bar{u_i}) > \lambda$, the following inequalities,
called the {\em effective capacity inequalities} are  valid and strengthen the flow-cover inequalities \cite{AardalPW95}.
%\begin{lemma}\cite{}
%Let $I \subset F$ be a cover with respect to $J$ and $J_i \subseteq J$, and assume that %$\max_{i\in I} (\bar{u_i}) > \lambda$.  The following effective capacity inequalities are valid.
\begin{center}
$\sum_{i\in I}\sum_{j\in J_i}d_jx_{ij}  + \sum_{i\in I}(\bar{u}_i-\lambda )^+(1-y_i)\leq d(J)$ 
\end{center}
\vspace*{-0.3cm}
%\end{lemma}
The proof of the following theorem uses some of the ideas we introduced
earlier for Theorem~\ref{cfl-SA:theorem}. In the appendix we give
Theorem~\ref{thm:submod} which strictly 
generalizes 
Theorem~\ref{effective-cap:theorem} 
to the  so-called submodular inequalities.

\begin{theorem}\label{effective-cap:theorem}
The integrality gap of the relaxation obtained from (LP-classic)  with
the addition of the effective
capacity inequalities is unbounded, even  for uniform
\cfl\ with unit demands.   
\end{theorem}
\vspace*{-0.2cm}
\begin{proof}
Consider an  instance with   $n$ cheap and $n+2$  costly facilities
and $Un+1$  clients, $U=n^3.$ Define the bad  solution $s$, similarly to
Section~\ref{SA-result}, s.t. 
 for  every $ch \in
Cheap,$   $co   \in  Costly,$   and      client  $j,$   $y_{ch}=1,
x_{chj}=\frac{1-\alpha}{n},$ $y_{co}=10/n^2,
x_{coj}=\frac{\alpha}{n+2}.$  
Recall that $\alpha=n^{-2}.$   We add a
set of $n+2$ facilities $a_{i}$, $1\leq i \leq n+2,$ all with $0$ opening
costs, on the same point at  distance $1$ from the rest (an instance of
the so-called \emph{facility location on a line}). In the bad solution
$s$ we  additionally set $y_{a_{i}}=1$ and $x_{a_{i}j}=0$  for all $i$
and for all clients $j$.

We will prove that in every  cover $I$ with respect to some client set
$J$ and to the $J_i$ client sets  for each $i$, there must always be a
number of at least $2n^3$ clients whose assignment variables  to some costly and
to some  $a_{i}$ do not appear  in the constraint. 
%% This  is because if
%% either $\bar{u}_i=U$  for each $i\in Costly$  or $\bar{u}_{a_i}=U$ for
%% each  $i\leq  n+2$  then  the  excess capacity  $\lambda  >  U$  since
%% $d(J)\leq Un+1.$ This  contradicts  the requirement that $\lambda <
%% U$. 
This  is because if,
$\bar{u}_i=U$  for each $i\in Costly,$  or, $\bar{u}_{a_i}=U$ for
each  $i\leq  n+2,$  then  the  excess capacity  $\lambda  >  U$  since
$d(J)\leq Un+1.$ This  contradicts  the requirement that $\lambda <
U$. 
So  there must  be a costly  facility $i_{co'}$ and  some facility
$a_{i'}$   such   that   for    the   corresponding   sets   we   have
$|J_{i_{co'}}|,|J_{a_{i'}}| < U$, and so there is a set $J^*$ of $2n^3
$ clients  whose assignments to those  two facilities do  not  appear in the
constraint. We exchange the values of $x_{i_{co'}j}$ and $x_{a_{i'}j}$
for all $j \in J^*$, leaving everything else the same, and we obtain a
solution  $s'=(y',x')$.   We  can  prove  similarly  to  the  proof  of  Lemma
\ref{assi-sym} that $s'$ is  a convex combination of integer solutions
and thus solution $s$ satisfies the inequality since the  parts
of $s$ and $s'$ visible to that inequality are the same.

We modify  the construction of  Lemma \ref{assi-sym} in  the following
way: facility  $a_{i'}$ is  opened $100\%$ of  the time but  is active
$1-\sum_{i\in  Costly}y'_{i}$ of  the time,  when none  of  the costly
facilities  are  opened.  When it  is  not  active,  the capacity  of  its
corresponding bin is $0$. When a costly other than $i_{co'}$ is opened
the  experiment is  the same  as  in Lemma  \ref{assi-sym}. If  costly
facility $i_{co'}$ is opened the  capacity of the corresponding bin is
$w^2_{co'}=\frac{\sum_{j}x'_{co'j}}{y'_{i_{co'}}}$  and the  capacity of
the  cheap  is  $\frac{|C|-w^2_{co'}}{n}$.  We randomly  select   some
$w^2_{co'}$ clients  that do not belong  to $J^*$ to be  tossed in the
bin of $i_{co'};$ we randomly
distribute  the balls corresponding  to the  remaining clients  to the
slots of  the cheap facilities. When  $a_{i'}$ is active,  and thus no
costly facility  is opened, the  capacity of the corresponding  bin is
$w^1_{a_{i'}}=\frac{\sum_j  x'_{a_{i'}j}}{1-\sum_{i\in  Costly}y'_{i}}$
and the capacity of the cheap is $\frac{|C|-w^1_{a_{i'}}}{n}$. We select
randomly  some  $w^1_{a_{i'}}$  clients  in  $J^*$  and  we  toss  the
corresponding balls in the bin of $a_{i'}$.  We randomly toss the
remaining balls to the slots of the bins of the cheap facilities.

Note that  the above experiment  induces a distribution  over feasible
integer solutions since  all the defined bin capacities  are less than
$U$ (this is by  the choice of the size of $J^*$)  and every client is
assigned to exactly one opened facility in each outcome.  We do not need
this distribution to be assignment-symmetric. Observe that the expected
vector   with  respect   to  the   latter  distribution   is  solution
$s'$. Finally, note that we  once again treated the capacities $w$ of
the bins  as
being integral.  For fractional bin capacities (which is
actually  always the case  for the  defined $w$'s)  we can  define the
experiment in a similar way to the proof of Lemma \ref{assi-sym}.  
\end{proof}

\section{Proper Relaxations}\label{sec:firstfamily}

In this section  we present the family of  proper relaxations
and characterize their strength.
Consider a   $0$-$1$ $(y,x)$ vector on the set of
 variables  of the  classic  relaxation (LP-classic)
such that $y_i \geq x_{ij}$ for all $i \in F, j \in C.$  The meaning  of
 $y_i=1$ is the usual one that
 we open facility $i.$  Likewise, the meaning of $x_{ij}=1$ is
 that we assign client $j$ to facility $i$. We call such a vector a 
 \emph{class}. Note that  the definition is quite general  and a class
 can be defined  from any such $(y,x)$, which may or  may not have a
 relationship  to a  feasible  integer solution.  
%% Classes generalize the notion of a star. 
We  denote the  vector
 corresponding  to a  class $cl$  as $(y,x)_{cl}$.  We  associate with
 class  $cl$   the  {\em cost   of  the  class}   $c_{cl}=\sum_{i \mid
   y_i=1  \in
   (y,x)_{cl}} f_i+  \sum_{i,j \mid  x_{ij}=1 \in  (y,x)_{cl}}
 c_{ij}$. Let the {\em assignments of class} $cl$ be defined as 
 $Agn_{cl}= \{(i,j) \in F\times C \mid  x_{ij}=1$ in $(y,x)_{cl}\}$.
We say that $cl$ {\em contains}  facility $i,$ if the corresponding entry
$y_i$  in the vector $(y,x)_{cl}$ equals $1.$  
The set of facilities contained in $cl$ is denoted by 
 $F(cl).$

%%%%%%%%% FOCS Definition, lots of space 
\iffalse ------------

\begin{definition}  {\bf (Constellation LPs)} \label{def:constell}
Let $\mathcal{C}$ be a set of classes defined for an instance $I(F,C)$
of
\lbfl. Let $x_{cl}$  be a variable associated with class $cl \in
\mathcal{C}.$
The {\em  constellation LP with class set}
  $\mathcal{C}$ is defined as 

\begin{align*} 
\min \sum_{cl \in \mathcal{C}} c_{cl}x_{cl} &&  \tag{LP($\mathcal{C}$)} \\
\sum\nolimits_{{cl} \mid  \exists i:(i,j) \in  Assignments_{cl}}
x_{cl}=1 && \forall j \in C   \\
\sum\nolimits_{{cl} \mid  i  \in  F({cl})} x_{cl} \leq 1 &&
\forall i \in F  \\
 x_{cl} \geq 0  && \forall cl \in \mathcal{C} 
\end{align*} 

\end{definition}
We will refer simply to a constellation LP when
$\mathcal{C}$ is implied from the context.  
------------------ \fi 

\vspace*{-0.05cm}
\begin{sloppypar}
\begin{definition}  {\bf (Constellation LPs)} \label{def:constell}
Let $\mathcal{C}$ be a set of classes defined for an instance $I(F,C)$
of \cfl\ or \lbfl. Let $x_{cl}$  be a variable associated with class $cl \in
\mathcal{C}.$
The {\em  constellation LP with class set}
  $\mathcal{C},$ denoted LP($\mathcal{C}$),  is defined as $\min \{
\sum_{cl \in \mathcal{C}} c_{cl}x_{cl} \mid \sum_{{cl} \mid  \exists
  i:(i,j) \in  Agn_{cl}} x_{cl}=1  \; \forall j \in C, \:\;\; \sum_{{cl} \mid
  i  \in  F({cl})} x_{cl} \leq 1 \; \forall i \in F,  \:\;\;  x_{cl} \geq 0
\; \forall cl \in \mathcal{C} \}$.

\begin{comment}
\[ \min \{ \sum_{cl \in \mathcal{C}} c_{cl}x_{cl}  \mid \sum_{{cl} \mid  \exists i:(i,j) \in  Agn_{cl}}
x_{cl}=1  \; \forall j \in C,   \;\;\; 
\sum_{{cl} \mid  i  \in  F({cl})} x_{cl} \leq 1 
\; \forall i \in F,  \;\;\;  
 x_{cl} \geq 0   \; \forall cl \in \mathcal{C}  \}
\]
\end{comment}
\end{definition}
\end{sloppypar}
\vspace*{-0.03cm}

\noindent 
We  refer simply to a {\em constellation LP} when
$\mathcal{C}$ is implied from the context.  
We define the \emph{projection} $s'=(y^{s'}, x^{s'})$ of solution $s=(x^s_{cl})_{cl \in \mathcal{C}}$ 
of  LP$(\mathcal{C})$ to the  facility opening
and assignment variables $(y,x)$ as $y_i^{s'}=\sum_{cl|i\in cl}x_{cl}^{s}$ and 
$x_{ij}^{s'}=\sum_{cl| (i,j) \in Agn_{cl}}x_{cl}^{s}$.
%
\iffalse --------------- due to SPACE
We will restrict our attention to  constellation LPs that satisfy a  natural property:   the LP is symmetric
 with respect to the clients and  the  facilities. 
The fact  that all facilities have the  same capacity / lower bound and all
clients have unit demand makes  this  property quite sound. 
For a class   $cl$ and
$f_1: \{1,...,n \} \rightarrow \{1,...,n \}$ a permutation of the facilities, we denote by $cl_{f_1}$
the class resulting by exchanging  for all $k, j$ the values  of the $y_{k}$  and
$x_{kj}$ coordinates 
of $(y,x)_{cl}$  with   the value of  the $y_{f_1(k)}$  and $x_{f_1(k)j}$ coordinates of
$(y,x)_{cl}$. Similarly,  for $f_2: \{1,...,  m\} \rightarrow \{1,...,
m\}$ a permutation of the clients, we denote  by $cl_{f_2}$ the class resulting  by exchanging for
every $i$  the value of  the $x_{ik}$ coordinate of  $(y,x)_{cl}$ with
 the value of the $x_{if_2(k)}$ coordinate of $(y,x)_{cl}$.

\begin{definition} {\bf ($P_1$: Symmetry)} \label{def:symmetry}
We  say  that property  $P_1$  holds for the constellation linear program LP($\mathcal{C}$)   if  the
following is  true: let $\phi  (n)$ be any  permutation of $F$  and $\mu
(m)$ any permutation of $C$.
 Then, for every  class $cl \in \mathcal{C},$  $cl_{\phi}$ 
and $cl_{\mu}$ are also in $\mathcal{C}.$
\end{definition}
\vspace*{-0.6cm}

------------------------------ \fi 
%
We  restrict our attention to  constellation LPs that satisfy a
symmetry property that is very natural for uniform capacities and unit
demands. 

\vspace*{-0.1cm}
\begin{definition} {\bf ($P_1$: Symmetry)} \label{def:symmetry}
We  say  that property  $P_1$  holds for the constellation linear program LP($\mathcal{C}$)   if  
 for every  class $cl \in \mathcal{C},$  all classes resulting from
 a  permutation that relabels the facilities and/or the clients of
 $cl$ are
 also in $\mathcal{C}.$
\end{definition}
\vspace*{-0.3cm}

%The second property we require is the obvious one that the 
%relaxation  is {\em valid,} i.e., the projection of its  feasible region to 
%$(y,x)$ contains all the characteristic vectors of the feasible integer solutions of the instance.
\begin{definition}  {\bf (Proper Relaxations)}  \label{def:proper} 
We call {\em proper relaxation}  for \cfl\ (\lbfl\/)  a constellation LP
 that is valid and satisfies property $P_1.$ 
\end{definition}
\vspace*{-0.1cm}

\noindent 
 A simple  example of a constellation LP is the well-known {\em
  (LP-star)}  (see, e.g., \cite{JainMMSV03}) where $\mathcal{C}$
corresponds to the set of all  {\em stars}: 
a facility and a set of at most $U$ (or at least $B$ for \lbfl) clients assigned to
it.
Obviously (LP-star) is a proper relaxation, while (LP-classic) is equivalent to
(LP-star). Therefore proper relaxations generalize the known natural
relaxations for \cfl\ and \lbfl. 
\begin{comment}
Our  result on proper relaxations is that  proper LPs that  are not ``complex'' enough have an unbounded integrality gap while those  that
are sufficiently ``complex'' have an integrality gap of $1.$  To that end, we
define  the complexity  of  a  proper LP. 
\end{comment}
In order to  characterize the strength of a proper LP we need  the notion of
complexity.  
\iffalse
Furthermore, for each  such facility  $i$ we  denote by
$C_{cl}(i)$ the  set of clients  $j$ for which  there is a facility  $i$ so
that $x_{ij}=1$ in $(y,x)_{cl}$.
\fi 

\vspace*{-0.2cm}
\begin{definition} {\bf (Complexity of proper relaxations)}  \label{def:complexity}
Given an instance $I(F,C)$ of \cfl\ (\lbfl\/)
let $F'$ be  a 
maximum-cardinality set  of open facilities in an integral feasible
solution. The {\em complexity $\alpha$} of a  proper
relaxation $LP(\mathcal{C})$ for $I$ is
defined as the 
   $\sup_{cl \in \mathcal{C}} ({|F(cl)|}/{|F'|}).$
% $\sup_{cl \in \mathcal{C}} \frac{|F(cl)|}{|F'|}.$ 
\end{definition}

 \vspace*{-0.09cm}

%Note that for \lbfl\ it is possible  to have a proper relaxation with complexity greater
%than $1$. 
The  complexity of  a  proper LP  represents the  maximum
fraction of the  total number of feasibly openable  facilities that is
allowed in a single class. 
%% For   a  proper relaxation  $LP(\mathcal{C}),$ the  complexity
%% describes to what extent 
%%  classes in $\mathcal{C}$ consider the instance locally. 
A complexity of nearly $1$
means that there are classes that take each 
into consideration almost the whole instance
at once. Low complexity means that all classes consider
the assignments of a small fraction of the instance at a time.  
\iffalse ===========================
We remark
that    the   proper    LP    with   an    integral   polytope    from
Theorem~\ref{thm:gap1} has
a complexity of  $1$ since every class corresponds  by construction to
a feasible integral solution. 
(Clearly not every LP with complexity $1$ has an integrality gap of $1$
since it might contain weak classes together with the strong
ones.) 
============ \fi 
By increasing the complexity of a proper LP  for a given instance 
 we can produce strictly stronger 
proper relaxations, an example is given in the Appendix. 

\vspace*{-0.1cm}
\begin{theorem}\label{theorem:proper}
Every proper relaxation for uniform \cfl\ (\lbfl) with complexity $\alpha < 1$ has an
unbounded integrality gap. There is a proper  relaxation for
\cfl\ (\lbfl) of  complexity $1$  whose projection to $(y,x)$ expresses the integral polytope. 
\end{theorem}

\bibliographystyle{plain}

\bibliography{bibliography-ver1}

\appendix

%\section{Appendix to Section~\ref{sec:prel}}

\section{Appendix to Section~\ref{SA-result}}

Omitted part of the proof of Lemma~\ref{assi-sym}.

\begin{proof}
First we explain how to handle fractional bin capacities. 
To handle the case where the  $w$'s are not integers, we simply do the
following:  each  time  costly  facility $i_b$  ($i_{c}\neq  i_b$)  is
picked, we set the number of slots of the corresponding bin to $\lfloor
w^1_{i_b} \rfloor$  ($\lfloor w^2_{i_{co}} \rfloor$)  with probability
$1-(w^1_{i_b}-\lfloor  w^1_{i_b} \rfloor  )$ ($1-(w^2_{i_{co}}-\lfloor
w^2_{i_{co}} \rfloor )$), otherwise set the slots to $\lceil w^1_{i_b}
\rceil$($\lceil  w^2_{i_{co}}  \rceil$). If  the  number  of slots  of
$i_{b}$  ($i_{co}$) is  set to  $\lfloor w^1_{i_b}  \rfloor$ ($\lfloor
w^2_{i_{co}}   \rfloor$)  then   we  pick   some  $n(\frac{|C|-\lfloor
  w^1_{i_b}  \rfloor  }{n}   -  \lfloor  (\frac{|C|-\lfloor  w^1_{i_b}
  \rfloor }{n}) \rfloor)$  ( $n(\frac{|C|-\lfloor w^2_{i_{co}} \rfloor
}{n}   -  \lfloor   (\frac{|C|-\lfloor   w^2_{i_{co}}  \rfloor   }{n})
\rfloor)$)  cheap facilities  at  random and  set their  corresponding
number of  slots to $\lceil \frac{|C|-\lfloor  w^1_{i_b} \rfloor }{n}
\rceil$ ($\lceil \frac{|C|-\lfloor w^2_{i_{co}} \rfloor }{n} \rceil$)
and the number of slots of the rest of the cheap facilities to
$\lfloor \frac{|C|-\lfloor        w^1_{i_b}        \rfloor        }{n}
\rfloor$($\lfloor \frac{|C|-\lfloor    w^2_{i_{co}}    \rfloor    }{n}
\rfloor$).  Otherwise  pick some $n(\frac{|C|-\lceil  w^1_{i_b} \rceil
}{n} -  \lfloor (\frac{|C|-\lceil  w^1_{i_b} \rceil }{n})  \rfloor)$ (
$n(\frac{|C|-\lceil    w^2_{i_{co}}     \rceil    }{n}    -    \lfloor
(\frac{|C|-\lceil   w^2_{i_{co}}   \rceil   }{n})   \rfloor)$)   cheap
facilities at  random and set  their corresponding number of  slots to
$\lceil  \frac{|C|-\lceil  w^1_{i_b}  \rceil  }{n}  \rceil$  ($\lceil
 \frac{|C|-\lceil w^2_{i_{co}} \rceil }{n}  \rceil$) and the number of
slots of the rest to $\lfloor \frac{|C|-\lceil w^1_{i_b} \rceil }{n} \rfloor
$($\lfloor \frac{|C|-\lceil  w^2_{i_{co}} \rceil }{n}  \rfloor$). Note
than in every case the expected  number of slots per facility is as in
the previous experiment.

As for the expected assignments to the cheap
facilities, observe that in every outcome of the experiment the demand
not assigned  to costly facilities  is exactly the demand  assigned to
cheap.  Since we  have proved  that  the expected  assignments to  the
costly  facilities are  those of  the  bad solution,  by linearity  of
expectation we get that the  total assignments to all cheap facilities
are $\sum_{i  \in Cheap}\sum_{j}x_{ij}'$ (the total  assignment of each
client adds up  to $1$ by the constraints of the  LP). By the symmetric
way the  cheap are handled  in the experiment  we have that  the total
expected demand assigned to  each $i\in Cheap$ is $\sum_{j}x_{ij}'$ and
by  the symmetric  way the  clients are  assigned to  $i$  through the
experiment we get  that the expected assignment of each  $j$ to $i$ is
$x_{ij}'$.

\end{proof}

Proof of Proposition \ref{level-ratio}.

\begin{proof}
Since  the  considered  distribution  is  assignment-symmetric,  event
$\mathcal{E}_p$ is  equivalent to  the event of  randomly distributing
$l+1$ balls to the slots of $n+1$  bins, with at most one ball in each
slot, each bin having $w^1_{ch}$ identical slots, asking that ball $j$
is  tossed in  the bin  of $i_b$  and ball  $j_{r}$ is  tossed  in bin
$i_{r}$.  Since there are  $\Theta (n^3)$  slots in  each bin  and the
balls    are   at    most   $n$,    it   is    easy   to    see   that
$(1-o(1))\frac{P[\mathcal{E}_p]}{n}   \leq   P[\mathcal{E}_{p'}]  \leq
(1+o(1))\frac{P[\mathcal{E}_p]}{n}$.

\end{proof}

Proof of Proposition \ref{transf_fraction}.

\begin{comment}
\begin{proposition}\label{transf_fraction}
Let $\mathcal{E}_p:x_{i_bj}=1 \ \wedge x_{i_{a_1}j_{b_1}}=1 \wedge x_{i_{a_2}j_{b_2}}=1 \wedge ...x_{i_{a_{k-i+1}}j_{b_{k-i+1}}}=1$ and $\mathcal{E}_{p/fixed}:x_{i^*j}=1  \wedge x_{i_{a_1}j_{b_1}}=1 \wedge x_{i_{a_2}j_{b_2}}=1 \wedge ...x_{i_{a_{k-i+1}}j_{b_{k-i+1}}}=1$ with $i^*\in Costly-\{i_b\}$. Then 
%$\frac{P[\mathcal{E}_p]}{P[\mathcal{E}_{p/fixed}]}<e^2$.
$(1+1/n)^{k-i+1}\leq \frac{P[\mathcal{E}_p]}{P[\mathcal{E}_{p/fixed}]}\leq (1+2/n)^{k-i+1} <e^2$.
\end{proposition}
\end{comment}

\begin{proof}
Consider  again the  random  experiment of  the  proof of  Proposition
\ref{level-ratio}.     Recall   that,    ignoring    constant   factors,
$w^1_{ch}=n^3-1$ and $w^2_{ch}=n^3-n^2.$ $P[\mathcal{E}_p]=x_{i_bj}P[
  x_{i_{a_1}j_{b_1}}=1      \wedge     x_{i_{a_2}j_{b_2}}=1     \wedge
  \ldots \wedge x_{i_{a_{k-i+1}}j_{b_{k-i+1}}}=1     \mid    x_{i_bj}=1]$     and
$P[\mathcal{E}_{p/fixed}]=x_{i^*j}P[x_{i_{a_1}j_{b_1}}=1         \wedge
  x_{i_{a_2}j_{b_2}}=1  \wedge  \ldots \wedge x_{i_{a_{k-i+1}}j_{b_{k-i+1}}}=1  \mid
  x_{i^*j}=1]$ and since $x_{i_bj}=x_{i^*j}$  we can compute the ratio
of the  probability of  success of the  tossing of $k-i+1$  balls when
$x_{i_bj}=1$, and thus the capacity of the bins corresponding to cheap
facilities is $w^1_{ch}$, to the probability of success of the tossing
of $k-i+1$ balls  when $x_{i^*j}=1$ and thus the  capacity of the bins
corresponding to cheap facilities is $w^2_{ch}$. When tossing the ball
$j_{b_r}$ given the successful  tossing of balls $j_{b_q}$ with $q<r$,
the  probability   of  success  is   $\frac{w^1_{ch}-o}{|C|-r+1}$  and
$\frac{w^2_{ch}-o}{|C|-r+1}$  respectively, where  $0\leq o\leq  r$ is
the  number  of  balls  already   placed  in  some  slot  of  the  bin
corresponding    to    cheap   facility    $a_r$.    We   have    that
$(1+(1-o(1))1/n)\leq         \frac{w^1_{ch}-o}{w^2_{ch}-o}        \leq
(1+(1+o(1))1/n)$.
\begin{comment}
Note that again both $w^1_{ch},w^2_{ch}$ are $\Theta (n^3)$ and $k-i+1<n$. In the ball tossing experiment, the probability of success of ball $j_r$ in the case where the capacities of the cheap bins is $w^1_{ch}$ is at most $(1+2/n)$ the probability of success of the same event in the case where the capacities of the cheap bins is $w^2_{ch}$ - we are very generous here, in fact the ratio of probability of success is just a little over $1+1/n$.
\end{comment}
 So $(1+(1-o(1)1/n)^{k-i+1} \leq \frac{P[\mathcal{E}_p]}{P[\mathcal{E}_{p/fixed}]}\leq (1+(1+o(1))1/n)^{k-i+1} < e^2$ using that $\lim_{x \rightarrow \infty}(1+d/x)^{x}=e^d$. 
\end{proof}

Proof of Proposition \ref{transf:prop}.

\begin{proof}
The intuition behind the proof is that the ``donor'' event that
supplies the required measure is much more likely to occur than the
events that require the transfusion. 

Consider  the measure  $t$ in  $E_d$ of  the set  of  integer solutions
satisfying $y_{i_b}=1$ and all events encountered at any iteration 
being false,
namely  $x_{i_bj}=0  \wedge  x_{i_1j_1}=0 \wedge  x_{i_2j_2}=0  \wedge
\ldots \wedge x_{i_{k}j_{k}}=0$. Then, by the random experiment of the
construction  of $E_d$,  this event  is equivalent  to the  event that
facility  $i_b$   is  picked,  $x_{i_bj}=0$  and   the  $k$  balls
corresponding to the clients of the rest of the events are not tossed in
their  corresponding bins. Using  again that  both $w^1_{ch},w^2_{ch}$
are $\Theta (n^3)$ and $k<n$,  we can bound the probability of the
$k$ balls by that of  $k$ Bernoulli trials with probability of
success $2/n$ (we are once  again very generous). Then the probability
that all events %% on ball $j_i$ 
fail is at least $(1-2/n)^{k}> \lim_{n
  \rightarrow \infty }(1-2/n)^n = 1/e^2$. Thus measure $t$ is at least
$(y_{i_{b}}-x_{i_bj}) 1/e^2$ which is  constant. On the other hand the
measure   required   by   the   transfusion  step   for   each   event
$\mathcal{E}_{p}$ of iteration  $i$ that needs to be  fixed is at most
$(e^2-1)P[\mathcal{E}_{p/fixed}]=\Theta     (1/n^i)$.     There    are
$k+1\choose{k-i+1}$ such events of iteration $i$, and summing over all
the iterations of our  construction we get $\sum_{i=1}^{k} {k+1\choose
  k-i+1}   \Theta   (1/n^i)   $    which   quantity   is   less   than
$(y_{i_{b}}-x_{i_bj})  1/e^2$  for  the   $k=n/10$  levels  of  SA  we
consider, so we can always pick the required amount of measure.
%We did not make any efforts to optimize the constants.
\end{proof}

\subsection{SA gap for \lbfl} 

A  similar result  to Theorem~\ref{cfl-SA:theorem}  can be  proved for
\lbfl\/. Consider an instance with $n$ facilities, lower bound $B=n^3$
and a total of $n(B-1)$ clients. The metric space here is more intriguing
than the one for the \cfl\ case. Consider a regular
$(n-1)$-dimensional 
simplex with
edge length $1.$  On each of  the $n$ vertices of the simplex  a facility along with
some $B-1$  clients are  located. All opening  costs are  $0.$ Clearly
every integer solution has a cost  of at least $B-1$ since we can open
at most  $n-1$ of the facilities,  and so at least  $B-1$ clients will
have to  be assigned to some facility  other than the one  on the same
vertex.  We call a client $j$  that is located on the same vertex with
facility   $i,$  \emph{exclusive}   client   of  $i$.   We  denote   by
$Exclusive(i)$ the set of clients  that are exclusive to facility $i$.
On the other hand we can  show that the following bad solution $s$ is
feasible at  $\Omega(n)$ levels of the  SA hierarchy. 
 For  all $i\in  F,$ $y_{i}=1-n^{-2}$; 
 for  a client $j\in  C,$ $x_{ij}=1-10n^{-2},$  if
$j \in Exclusive(i),$ and
$x_{ij}=\frac{10n^{-2}}{n-1}$  for   all  other  facilities.  Solution
$s$ incurs a cost of $o(B)$.

\begin{theorem}\label{lbfl-SA:theorem}
There is  a family of  \lbfl\   instances with $n$ facilities and $n^4-n$ clients  
such that the  relaxations  obtained from (LP-classic)  at $\Omega (n)$
levels of  the Sherali-Adams hierarchy have an unbounded integrality
gap. 
\end{theorem}

The proof  is similar to  that of \cfl\  and is thus omitted.  Here the
reader can  find a  sketch of  the necessary changes  to the  proof of
Theorem~\ref{cfl-SA:theorem}.

\vspace*{0.2cm}
\noindent
{\em Sketch of proof of Theorem \ref{lbfl-SA:theorem}.}
Consider a constraint $\pi: \sum_{j}x_{i^{\pi}j}\geq By_{i^{\pi}}$ and
a  multiplier $z$  at  level $k$  and  let $v(\pi,z)$  be  the set  of
variables appearing in the  multiplied constraint.  We pick a facility
$i_b$ not  in $v(\pi,z)$  to take the  blame. We construct  a solution
$s'$ where we set $y'_{i_b}=n-1-\sum_{i\neq i_b}y_{i}$ and for each $j
\in  Exclusive(i_b)$ we  set  $x'_{i_bj}=y'_{i_b}=\frac{1-1/n}{n}$ and  we
distribute the remaining demand that was assigned to $i_b$ to each
facility from a constant-size set $I_b$ of facilities  not appearing in
$v(\pi,z)$. Solution $s'$  can be obtained as a  convex combination of
integer solutions  by constructing  a distribution similarly  to Lemma
\ref{assi-sym}. This time the distribution satisfies 
 that exactly  $n-1$ facilities are opened in
each  outcome of the  experiment. Note that we do not require the 
underlying distribution to be assignment symmetric, 
because facilities have to treat differently their exclusive clients. 
We  set the  
values of  the linearized products appearing in the multiplied constraint
equal to  the probability of the  corresponding events with respect  
to the aforementioned  distribution. No product involving variables of 
$i_b\cup I_b$ appear in the constraint.  For  constraints 
$0\leq  x_{ij},y_i \leq 1$  and $x_{ij} \leq y_{i}  $ the  construction 
of the  distribution is the  same. The distributions constructed so far
 are locally consistent as required by Lemma \ref{SA-survival}.

The case where the constraint is $\pi: \sum_{i}x_{ij^{\pi}}=1$ is once
again  more complicated.   We choose  a  facility $i_b  \notin z$  and
moreover $j^{\pi}\notin Exclusive(i_b)$ to  take the blame and the set
$I_b$ is  defined as before except  we also require  that $j^{\pi}$ is
not exclusive to any of them. Solution $s'$ is constructed like in the
previous case. All products take the value of the corresponding events
in the  distribution except those in which the  unique variable involving
$i_b$  appears, namely  $x_{i_bj}$ and  those involving  facilities in
$I_b$. We perform a transfusion  step so that the probabilities of all
the  events   whose  corresponding  products  appear   in  the  lifted
constraint become consistent with the distributions of the previous
case:  this  time we  need  to fix  the  probabilities  of the  events
involving facility $i_b$ or some facility $i\in I_b$.

\subsection{Robustness of the SA gap} \label{sec:robust} 

In  this section  we explain to the  interested reader  how
adding simple valid inequalities does  not affect our arguments on the
SA hierarchy. 

As an example we address the valid inequality 
$\sum_i y_i  \geq \lceil D/U  \rceil$, where $D$  is the
total  amount  of  demand.  This  is  a  well-known
facet-inducing constraint for our instance, see, e.g., %% (31) in
\cite[p. 283]{LeungM89}. 
Of course this inequality is  rendered useless by  slight modifications to
the instance and the bad solution. 
Identifying  ``areas'' of a fractional  solution  where  the demand  exceeds the
available capacity is impossible  without some yet unknown form of preprocessing.
In fact part of the motivation behind   Theorem~\ref{cfl-SA:theorem}
is to demonstrate that the 
SA hierarcy is inadequate for such preprocessing purposes. It therefore 
suffices  to include in the body of the paper 
the simplest possible proof for the theorem.

We modify  the family of  "bad" instances by  using the same  trick we
used in the proof  of Theorem~\ref{effective-cap:theorem}: we have $n$
cheap  and $n$  costly  facilities  and $Un+1$  clients,  and the  bad
solution  in which  for every  $ch \in  Cheap,$ $co  \in  Costly,$ and
client    $j,$     $y_{ch}=1,    x_{chj}=\frac{1-\alpha}{n},$    $y_{co}=10/n^2,
x_{coj}=\frac{\alpha}{n}$ and additionally we add a set of $n$ 
{\em dummy} facilities
$a_{i}$, $1\leq  i \leq n,$  all with $0$  opening costs, on  the same
point  at distance  $1$ from  the  rest. In  the bad  solution $s$  we
additionally set $y_{a_{i}}=1$ and  $x_{a_{i}j}=0$ for all $i$ and for
all clients $j$. The inequality is obviously satisfied.

In the  design of  the locally consistent  distributions, now  we must
give a distribution for the case where the constraint $\pi$ is the new
one $\sum_i y_i \geq \lceil D/U \rceil$, and verify that the "visible"
part of  the distribution  agrees with the  visible part of  all other
distributions of  the proof.   In this case  there must be  some dummy
facility $a_d$ not appearing as an  index in the multiplier $z$ of the
constraint    (although   its    $y$   variable    does    appear   in
$\pi$). Additionally  there must be  a costly facility $i'$  for which
the assignments  of clients to  $i'$ do not  appear in $v(\pi,  z)$ --
this is  ensured by the number  of rounds we consider.   We modify the
solution $(y,x)$  to obtain $(y',  x')$ where the facilities  $i'$ and
$a_d$  exchange the  values  of their  corresponding assignments.   We
define  now the  random experiment  similarly  to the  proof of  Lemma
\ref{assi-sym}  with  facility  $a_d$   taking  the  blame.  The  only
difference is that while $a_d$ is  opened 100\% of the time, it is not
assigned any demand when a  costly facility other than $i'$ is opened.
In  the terminology  of  Theorem~\ref{effective-cap:theorem}, $a_d$  is
always open but  it is inactive when some $i \in  Costly,$ $i \neq i',$
is  opened.  It  is  easy to  see  that the  distribution obtained  is
consistent with all the  other distributions defined for this modified
instance, as required by Lemma \ref{SA-survival}.

\section{Appendix to Section~\ref{flow-cover}}

\subsection{How to fool submodular inequalities} 

Here we show that the classic relaxation strengthened by the
submodular inequalities has unbounded gap. The submodular inequalities
introduced in \cite{AardalPW95} are even stronger than the effective
capacity inequalities. We limit our discussion to uniform \cfl\ where
all clients have unit demands. 

Choose a subset $J\subseteq C$ of  clients, and let $I \subseteq F$ be
a subset  of facilities.  For each  facility $i\in I$  choose a subset
$J_i \subseteq J$.  Consider a 3-level network $G$ with a source $s,$ a set of
nodes corresponding to the facilities, a set of nodes corresponding
to the clients and a sink $t$. The source $s$ is connected by an edge of capacity
$\min\{U, |J_i|\}$ to each facility node $i.$ That node 
 is connected by an  edge of unit
capacity to each node corresponding to client $j,$ $j\in J_i$.
Each node corresponding to some client is connected by an edge of unit
capacity to the sink $t$.
%Once the sets K, J and Kj for j 6 J are
%known, we can define the effective capacity of depot j as rhj = min(m~, d(Kj)).

Define   $f(I)$ as the maximum $s$-$t$ flow value in $G.$ 
%% from the facilities in $I$ to the clients in $J$
%% given the arc set $\{(j,k): i \in I, j \in J_i\}$. 
Define $f(I\setminus\{i\})$ as the maximum flow when facility $i$ is
closed, i.e., when the
capacity of edge $(s,i)$ is set to zero. 
The difference in maximum flow when all facilities in $I$
are open, and when all facilities except facility $i$ are open, is
called the 
{\em increment}
function and is defined as $\rho_i(I \setminus \{i\}) = f(I) - f(I
\setminus \{i\})$.

For any choice of $I \subseteq F,$ $J\subseteq C,$ and $J_i \subseteq
J,$ for all $i,$ the following inequalities, 
 called \emph{the submodular inequalities,} are valid for
 \cfl\ \cite{AardalPW95}. The name reflects the fact that the function
 $f(I)$ is submodular. 

\begin{center}
$\sum_{i\in I}\sum_{j\in J_i}x_{ij}+\sum_{i\in I}\rho_i(I\setminus \{i\})(1-y_i) \leq f(I)$
\end{center}

\begin{theorem}  \label{thm:submod} 
The integrality gap of (LP-classic) remains unbounded even after the addition of the submodular inequalities.  
\end{theorem}

\begin{proof}
Consider the  instance and  bad solution  $s$ that we  used in
Theorem~\ref{cfl-SA:theorem} 
for the SA
result.  To prove  that $s$  is  feasible for  the classic  relaxation
strengthened  by  the submodular  inequalities  we  take  the idea  of
fooling local  constraints a little further: either  the constraint is
local  enough that  we  can use  the  ideas from  our previous  proofs
(define $s'$ that is a convex combination of integer solutions and has
the same  visible part as $s$  with respect to the  constraint), or we
can  define another  instance $I'$  and solution  $s'$ for  which the
inequality in question is true with respect to $s'$ and again $s'$ has
the same visible part as $s$ with respect to the constraint. Note that
our arguments  include two different  instances as opposed to  all our
other proofs so far.

Consider the submodular inequality $\pi$ for some $I,$ $J$ and some
selection of $J_i$'s.  If not all the costly  facilities appear in the
constraint the  proof is similar  to that of Lemma  \ref{assi-sym}. If
at least  $n$ assignment variables
 to cheap facilities do not appear in $\pi$ we do the
following: we add  one more facility $a$ to  the instance. 
We construct  a solution $s'$ for the new instance
$I'$ as follows. 
We transfer
the demand  corresponding to the  missing assignments of the  cheap to
$a,$  and we set $y_a=1.$
Observe that $\pi$ is valid for $I'.$ 
%% and (ii) $a$ is
% assigned fractionally at least one unit of demand. N
Now we
can  show that  $s'$  is  a convex  combination  of integer  solutions
similarly to the proof of Theorem \ref{effective-cap:theorem}, where
the role of $J^*$ is played by those clients whose assignments were
transferred from the $Cheap$ to $a.$ 
Facility  $a$  will be  active  only  when  no costly  facilities  are
open. Because, in the fractional solution $s',$
$a$ is assigned a total demand of at least $1-1/n^2,$ in each outcome of the
random experiment in  which $a$ is active, it will  be assigned at least
one client. 
 By the convex combination produced, 
the inequality is satisfied by  $s'$. Thus the same inequality for the
original instance is satisfied by $s.$

Now  consider  the case  where  less  than  $n$ assignments  to  cheap
facilities are missing from $\pi.$ We will show that it cannot be the case that
all $y_i$ variables  of costly facilities appear in  the constraint as
well.  Consider  the quantity  $\rho_i(I\setminus \{i\})$  for  some costly  facility
$i$. If $\rho_i(I\setminus \{i\})>0,$  
then $J_i$ is not empty. We will show that   the set of nodes
$(Cheap \cap I) \cup \{i\}$ in $G$ has enough incident edges so that the flow
originating from them is equal to  the total client demand $|J|$ in $G.$
We first give some properties of graph $G.$

\begin{claim}
If less than $n$ assignments to cheap
facilities are missing from $\pi,$ then $(Cheap \cap I) = Cheap$ and 
$J=C.$ 
\end{claim}

\noindent{\em Proof of Claim.} 
To see that $(Cheap \cap I) = Cheap,$ notice that if a cheap
facility  is  missing  from  $I,$  at least  $|C|=n^4  +1$  assignment
variables will
be missing  from $\pi,$  a contradiction. For  the second part  of the
claim, if a client $j$ is missing from $J,$ then all the corresponding
$n$ edges that would connect $j$ to a cheap facility cannot be in $G.$
Therefore at least $n$ assignment-to-cheap variables are missing from $\pi,$ a
contradiction. The proof of the claim is complete.

We return to proving that $Cheap \cup \{i\}$ has enough incident edges so that the flow
originating from them is equal to  the total client demand $|C|$ in $G.$
``Assign'' one client $j\in J_i$ to facility $i$ and for the remaining
 $|C|-1$ clients do the following: assign each client $j'$
involved in the set of variables  of assignments-to-cheap that are
missing from $\pi$ to a
cheap facility $i'$  such that $j'\in J_{i'}$. There  is always such a
cheap facility $i'$  since the missing edges from  the client-nodes in
$G$ to the cheap-facility nodes are less than $n.$
%% This means that each facility node in $(Cheap \cap I)$ has degree at
%% least 
Assign
the remaining clients arbitrarily to the cheap facilities respecting the
capacities,  since all  the  edges  from cheap  to  those clients  are
included in the network. Thus it must be the case that 
$\rho_{i'}(I\setminus \{i'\})=0$
for any other costly facility $i' \neq i$. Since the $y_{i'}$ variable of
such a  facility $i'$  has $0$ coefficient  in the constraint,  it can
take  the   blame  and  the  proof   is  similar  to   that  of  Lemma
\ref{assi-sym}.
\end{proof}

\section{Appendix to Section~\ref{sec:firstfamily}}

\begin{example}\label{proper_str}
An increased complexity allows strictly stronger proper relaxations.
\end{example}

First we show how one can construct any integer solution using classes that open the
same number of facilities.
Consider an integer solution $s$ with opened facilities $1,\ldots ,t$. We will use the following classes 
in which exactly $r<t$ facilities are opened:
For any set of  $t$ consecutive classes in a cyclic ordering, namely $(1,\ldots ,r),(2,\ldots ,r+1),\ldots ,(t,\ldots ,r-1)$, define a class that opens those facilities and makes the same assignments to them 
as $s$. Then the integer solution is obtained  if for every $cl$ we set $x_{cl}=1/r$.
Observe that the latter solution is feasible for the proper relaxation.

We give a toy example showing that by increasing the complexity, we can
get strictly stronger relaxations. Consider an \lbfl\ instance with $4$ facilities $2$ sets $S_1,S_2$
of $13$ clients each and 2 sets $S_3,S_4$ of $9$ clients each and $B=10$.  For the star relaxation
(complexity $\alpha=1/4$ for this instance)
there is a feasible solution $\bar{s}$ whose projection to $(y,x)$
 is the following $(\bar{y},\bar{x})$: for facility $1,$ $\bar{y}_1=1$ and is assigned $S_1$ integrally, for facility $2,$ $\bar{y}_2=1$ and is assigned $S_2$ integrally, for facility $3,$ $\bar{y}_3=9/10$ and is assigned each client of $S_3$ with a fraction of $9/10$ and each of $S_4$ with $1/10$, and similarly for facility $4,$ $\bar{y}_4=9/10$ and is assigned
each client of $S_4$ with a fraction of $9/10$ and each of $S_3$ with $1/10$. Actually
a direct consequence of Theorem \ref{theorem:proper} is that for any proper relaxation of the same complexity as the star relaxation, the above solution is feasible.

Now consider the following proper relaxation: all characteristic vectors 
of integer solutions with at most
$3$ facilities are classes plus all the 
vectors of solutions with $4$ facilities restricted in any $3$ facilities ($3/4$ parts of integer solutions that open all four facilities).
It is symmetric and valid by the previous discussion and has complexity $\alpha=3/4$. 
In any assignment of values to the class variables  that projects to $(\bar{y},\bar{x})$ the following are true:
since classes with less than $3$ facilities are integer solutions, they contain
assignments for all the clients and thus if we were 
to use a non-zero measure of such classes we would make non-zero assignment 
that does not exist in the support of $(\bar{y},\bar{x})$.
 If we use
classes with exactly $3$ facilities, then exactly one of facilities $3,4$ must be present, 
since no integer solution opens them both with just the clients in $S_3 \cup S_4$. 
So we have to use at least $\bar{y_3}+\bar{y_4}=18/10$ measure of such classes. 
So each one of facilities $1,2$
must be present in more than a unit of classes, which would make the solution infeasible.

\vspace*{0.8cm}

\noindent
{\em Proof of Theorem \ref{theorem:proper}.}

We first prove the easy part, 
that there are proper relaxations for \cfl\ and \lbfl\ with complexity $1$ that
express the integral polytope.
For a given instance let $\mathcal{C}$ consist of a class for each
distinct integral solution. The resulting $LP(\mathcal{C})$ is clearly
proper. Let $x$ be any feasible solution of $LP(\mathcal{C})$ and let
$S$ be the support of  the solution. For every $cl \in S,$ and 
for every client $j \in C,$ there is an $i \in F,$ such that $(i,j)
\in Assignments_{cl}.$ Therefore 
$$ \sum_{cl \in S} x_{cl} = 1.$$
This implies that $x$ is a convex combination of integral
solutions. By the boundedness of the feasible region of
$LP(\mathcal{C}),$ the  corresponding polytope is integral.  
Clearly not every LP with complexity $1$ has an integrality gap of $1$
since it might contain weak classes together with  strong
ones.

In the next two subsections, 
we prove the first part of Theorem \ref{theorem:proper} for \lbfl\ and \cfl\ respectively.

\subsection{Proof of Theorem \ref{theorem:proper} for \lbfl\ }
\label{sec:proof_theorem_p1}

Our proof includes the following steps. We define an instance $I$ 
and consider any proper relaxation $LP(\mathcal{C})$ for $I$ that has complexity
$\alpha <1.$ 
Given $\alpha,$ we use   the validity  and symmetry properties to show the existence of
a specific set of classes in $\mathcal{C}$. Then we use these classes to construct a
desired feasible fractional solution, relying again on symmetry. 
In the last step  we specify  the distances between the clients and  the facilities, so
that the instance is metric and the constructed solution has an  unbounded integrality
gap.

\subsubsection{Existence of a certain type of classes}

Let us fix for the remainder of the section 
an instance $I$ with $n+1$ facilities, where $n$ is
sufficiently large to ensure  that $\alpha n \leq n - c_0$  where
$c_0,$  is a
constant greater than or equal to  $2$. Let the bound $B=n^2$, and let
the number of  clients be $n^3$. Notice that  there are enough clients
to open $n$ facilities, with  exactly $n^2$ clients assigned  to each
one that is opened. The  facility costs  and the assignment  costs will  be defined
later.  Recall  that the  space  of  feasible  solutions of  a  proper
relaxation is independent of the costs.

 We  assume that  the  facilities are  numbered
$i=1,2,\ldots ,n+1$. 
For a solution $p$ we  denote by $Clients_p(i)$
the set  of clients  that are assigned  to facility $i$  in solution
$p$, and  likewise for a  class $cl$ we denote  by $Clients_{cl}(i)$
the set of clients that are assigned to facility $i$.  
Consider  an  integral  solution  $s$  to  the  instance  where 
facilities $1,\ldots ,n$  are opened. 
Since our proper  relaxation is valid, it must have   a feasible 
solution $s'=(x_{cl})_{cl \in \mathcal{C}}$  whose projection to $(y,x)$ gives the characteristic
vector of $s$.  We prove the existence of a  class $cl_0,$ with some desirable
properties, in the support of $s'.$ 

By Definition~\ref{def:constell},
$s'$ can only be obtained as a positive combination of classes $cl$ such that for
every    facility   $i$    we   have    $Clients_{cl}(i)   \subseteq
Clients_s(i)$, Otherwise,  if the variables  of a  class $cl$
with $Clients_{cl}(i)  \setminus  Clients_s(i) \neq \varnothing$ have  non-zero value,
then in $s'$ there will be  some client assigned to
some facility with a positive fraction, while the projection of $s',$ namely $s,$
does  not include  the
particular  assignment.  
Moreover,  since exactly $B$ clients are  assigned to each
facility in  $s$,   for every facility $i$
that   is  contained   in  such   a  class   $cl,$  $Clients_{cl}(i)=
Clients_s(i)$. To see why this  is true, 
since in  $s$ we have $y_i=1,$ for all $i\leq n,$   it follows  that for every facility $i\leq n$,
 $\sum_{cl  | \exists (i,j)  \in  Agn_{cl}} x_{cl}  =1$.  
But  then  we have  that
$|Clients_s(i)|=B=\sum_{cl  |   \exists  (i,j)  \in  Agn_{cl}}
x_{cl}|Clients_{cl}(i)|$.  We have already established   that $x_{cl}>0
\implies |Clients_{cl}(i)|  \leq B$. Then $B$ is  a convex combination
of quantities less than or equal to $B$, so for all such classes $cl$
we have $|Clients_{cl}(i)|=B$.

Therefore
in the class set of any proper relaxation for $I,$ there is 
a class $cl_0$ that assigns exactly $B$ clients to each of  the
facilities in $F(cl_0).$ By the value of $\alpha,$ 
 $|F(cl_0)| \leq n -c_0.$ The following
lemma has been proved.

\iffalse     --- old version 
\begin{lemma}
There is a class $o$ that is contained in the class set of the proper
relaxation, that assigns $B$ clients to each of $n-c$ for some facilities.
\end{lemma}
\fi 
\begin{lemma}   \label{lemma:existence} 
Given the specific instance $I,$ any proper relaxation  of complexity
$\alpha$ for $I$ contains in its class set a class 
$cl_0$ that assigns $B$ clients to each of $n-c$ facilities, for some
integer  $c \geq 2.$ 
\end{lemma}

\subsubsection{Construction of a bad  solution}  
\label{subsec:badlbfl}

In the present section we will use the class $cl_0$ along with the
symmetric classes to construct a solution to the proper LP with 
the following
property: there are some  $q$ facilities   that
are almost integrally opened while the number of distinct  clients assigned to them will be less than $Bq$. 

Recall that by property $P_1$ every class that is isomorphic to $cl_0$ is
also a class of our proper relaxation. This means that
every set  of $n-c$ facilities and every  set of $B(n-c)$ clients
assigned to those facilities so that each facility is assigned exactly
$B$ clients, defines a class, called {\em admissible,} that belongs to the set of classes
defined of a  proper relaxation for the instance $I$.

Let  us  turn  again  to  the solution  $s$  to  provide  some  more
definitions.  For   every  facility  $i,$   $i=1,\ldots, n-1$,  we  choose
arbitrarily a client $j'$ assigned to  it by $s$. For each such facility
$i$  we   denote  by  $Exclusive(i)$   the  set  of   clients  $
Clients_s(i) - \{j' \},$ i.e., the set of clients assigned to
$i$ by $s$ after we discard $j'$ (we will also call them the
{\em exclusive clients of $i$}). For facilities $n,$ $n+1$ the sets
$Exclusive(n),$ $Exclusive(n+1)$ are identical  and defined to be equal to 
 the union of $Clients_s(n)$ with all
the  discarded clients from  the other  facilities. In  the fractional
solution that we will construct below, the clients in $Exclusive(i)$
will be almost integrally assigned to $i$ for $i=1,\ldots ,n-1$.

We  are   ready  to  describe  the  construction   of  the  fractional
solution. We will use a subset $S$ of admissible classes that 
do not contain both $n$ and  $n+1$. $S$ contains all such classes   
 $cl$ that assign to each facility $i \leq
n-1$  in the class  the set  of clients  $Exclusive(i)$ plus  one more
client selected  from the sets $Exclusive(i')$  for those facilities
$i' \leq n-1$ that do not belong  to $cl$ (there are at least $c-1$
of them). As for facility $n$ (resp. $n+1$), if it  is contained in $cl,$ then
it is  assigned some set  of $B$ clients  out of the total  $B+n-1$ in
$Exclusive(n)$ (resp. $Exclusive(n+1)$).  
All classes not in $S$  will get a value
of zero in our solution.
We
will distinguish the classes in $S$ into two types: the classes
of {\em type $A$} that contain facility  $n$ or $n+1$ but not both, and classes
of {\em type $B$} that
 contain neither $n$ nor $n+1$.

We consider  first classes of type  $A$. We give  to each such class   a
very small  quantity of  measure $\epsilon$. Let  $\phi$ be  the total
amount of measure used. We call this step $Round_A$.  The
following  lemma shows  that after  $Round_A$, the  partial fractional
solution  induced  by  the  classes  has a  convenient  and  symmetric
structure:

\begin{lemma}  \label{lemma:roundA}
After  $Round_A$,  each client  $j  \in  Exclusive(i),$  $i\leq n-1,$  is
assigned to $i$  with a fraction of $\frac{n-c-1}{n-1}  \phi$ and is
assigned to each other facility $i',$ $i' \neq i,$ $i' \leq  n-1,$ with a
fraction of $\frac{n-c-1}{(n-1)(n-2)(n^2-1)} \phi$. Each client $j \in
Exclusive(n)$ ($=Exclusive(n+1)$)   is   assigned   to   $n$ and to $n+1$ 
 with   a   fraction   of $\frac{n^2}{2(n^2+n-1)} \phi$.
\end{lemma}

\begin{proof}
Consider a facility  $i, i\leq n-1$. Since exactly one of facilities  $n,n+1$ is present in
all the classes of type  $A$ and each class contains $n-c$ facilities,
$i$ is  present in the  classes of $Round_A$  $\frac{n-c-1}{n-1}$ of
the time due to symmetry of the classes. Each time $i$ is present in
a class  $cl$ that  class $cl$ assigns  all $j \in  Exclusive(i)$ to
$i$.  So  client  $j$  is  assigned  to $i$  with  a  fraction  of
$\frac{n-c-1}{n-1}  \phi$. When $i$  is not  present in  class $cl$,
which happens  $\frac{c}{n-1}$ of the time, then  its exclusive clients
along with  the exclusive  clients of all  the other  $c-1$ facilities
that  are also  not  present in  $cl$  are used  to  help the  $n-c-1$
facilities $i \leq n-1,$  reach the bound $B$ of clients (recall
that the number of exclusive clients of each such facility is equal to
$B-1$).  Each time  this happens, the $n-c-1$ facilities  in $cl$ need
$n-c-1$  additional clients, while  the exclusive  clients of  the $c$
facilities that are  not present in $cl$ are  $c(n^2-1)$ in total. Due
to symmetry  once again, a  specific client $j \in  Exclusive(i)$ is
assigned to  one of those  $n-c-1$ facilities $\frac{n-c-1}{c(n^2-1)}$
of the  time of those cases.  So in total  this happens $\frac{c}{n-1}
\times  \frac{n-c-1}{c(n^2-1)}  =  \frac{n-c-1}{(n-1)(n^2-1)}$ of  the
time, so it follows that client $j$ is assigned to a specific facility
$i',$ $i' \neq i,$ $i' \leq n-1,$ $\frac{n-c-1}{(n-1)(n-2)(n^2-1)}$ of the
time. The fraction with which  $j$ is assigned to $i'$ after $Round_A$
is $\frac{n-c-1}{(n-1)(n-2)(n^2-1)} \phi$.

For  the proof  of the  second part  of the  lemma,  consider facilities
$n,n+1$. Each one of those is present in the classes of type $A$ an equal 
fraction $1/2$ of the time. The
only clients that  are assigned to them are  their exclusive clients. Each
class $cl$ assigns exactly $B= n^2$ out of those $n^2+n-1$ clients. So,
due to symmetry, each client $j \in Exclusive(n)$ is present in
$cl$ $\frac{n^2}{n^2+n-1}$  of the time,  so $j$ is assigned to $n$ and $n+1$
with a fraction of $\frac{n^2}{2(n^2+n-1)} \phi$ to each.     
\end{proof}

Note that  after $Round_A$ each  facility $i, i  \leq n-1,$ has  a total
amount $ \frac{(n-c-1)B}{(n-1)} \phi$  of clients (since it is present
in a class  $\frac{(n-c-1)}{(n-1)}$ of the time and  when this happens
it is  given $B$ clients).  Similarly, facilities $n,n+1$  after $Round_A$
have a total amount $B\phi /2$ each.

Now we can explain the underlying intuition for distinguishing between
the two
types of classes.  The feasible fractional solution $(y^*,x^*)$
we  intend to  construct is  the following:  for each
facility $ i\leq n-1,$ its exclusive clients are assigned to it with
a fraction of $\frac{n^2-1}{n^2}$ each, while they are assigned with a
fraction of $\frac{1}{(n^2)(n-2)}$ to  each other facility $ i'
\leq  n-1$. As  for facilities  $n,n+1$, all  of their exclusive  clients are
assigned with a fraction of $1/2$  to each.  If  we  project  the solution  to  
 $(y,x)$, the $y$ variables will be forced 
to  take   the  values
$y^*_i=\frac{n^2-1}{n^2},$ for $i \leq n-1,$ and $y^*_n=y^*_{n+1}=\frac{n^2+n-1}{2n^2}$. Observe as we give some  amount of  measure to  $Round_A$,
 the  variables  concerning the
assignments to facilities $n,n+1$ tend to their intended values in the
solution we want to construct ``faster'' than the variables concerning the
assignments to the other facilities. This is because, by Lemma~\ref{lemma:roundA}
after $Round_A$ each exclusive client  of $n,n+1$ is assigned to each of them with
a fraction of $\frac{n^2}{2(n^2+n-1)} \phi$ which is $\frac{n^2}{n^2+n-1}
\phi$ of  the intended value. At the  same time, every
exclusive  client of  each other  facility is  assigned to  it  with a
fraction of $\frac{n-c-1}{n-1} \phi$ which is $\frac{\frac{n-c-1}{n-1}
  \phi}{\frac{n^2-1}{n^2}}$  of the  intended value.  For sufficiently
large  instance  $I$,  as  $n$  tends  to  infinity,  the  assignments
to $n$ and $n+1$ will reach their intended values while there will
be   some    fraction   of   every    other   client   left    to   be
assigned. Subsequently we have to use classes of type $B$, 
to achieve the opposite effect: the
variables  concerning the  assignments of  the first  $n-1$ facilities
should tend
to their intended values ``faster''  than those of $n$ and $n+1$ (since
$n$ and $n+1$  are 
not  present in  any of  the classes  of type  $B$,  the corresponding
speed will actually be zero).

We  proceed  with  giving  the  details  of  the  usage  of  type  $B$
classes. As before,  we give to each such class  a very small quantity
of  measure $\epsilon$.  Let  $\xi$  be the  total  amount of  measure
used. We call this step $Round_B$.

\begin{lemma}   \label{lemma:roundB}
After  $Round_B$,  each client  $j  \in  Exclusive(i),$  $i\leq n-1,$  is
assigned  to $i$  with a  fraction of  $\frac{n-c}{n-1} \xi$  and is
assigned to each other  facility $i',$ $i' \neq i,$ $i'  \leq n-1,$ with a
fraction of $\frac{n-c}{(n-1)(n-2)(n^2-1)} \xi$.
\end{lemma}

\begin{proof}
The proof  follows closely that of  Lemma~\ref{lemma:roundA}. A  facility $i, i
\leq n-1,$  is present in  a class of  type $B$ $\frac{n-c}{n-1}$  of the
time (since  $c \geq  2$ this  fraction is less  than $1$).  Each such
time, every  client $j  \in Exclusive(i)$ is  assigned to  it (again
this  is due  to the  definition  of classes  of type  $B$). So  after
$Round_B$,   $j$   is  assigned   to   $i$   with   a  fraction   of
$\frac{n-c}{n-1} \xi$.  
Also, when $i$  is  present  in a  class, it is assigned exactly one client
which is exclusive to a facility
not in the class. Since in total there are $(n-2)(B-1)$ such candidate clients,
and by symmetry, after round $B$ 
each one of them is picked an equal fraction of the time to
be assigned to $i$, we have that
each client $j$ is assigned to a facility for which $j$ is not
exclusive with  a fraction
$\frac{n-c}{(n-1)(n-2)(n^2-1)} \xi$.    
\end{proof}

\noindent
To  construct  the   aforementioned  fractional  solution $(y^*,x^*)$,  set  $\phi=
\frac{n^2+n-1}{n^2}$    and     $\xi=    (\frac{n^2-1}{n^2}-\frac{n-c-1}{n-1}\phi)
\frac{n-1}{n-c}$, and  add the  fractional assignments of  the two
rounds. 

It is easy to check that the facility and assignment variables of facilities $n,n+1$
take the value they have in $(y^*,x^*)$. Same is true for the facility variables for $i \leq n-1$
and the assignment variables of the clients to the facilities they are exclusive. 
To see that the same goes for the non-exclusive assignments, observe that since
every class assign exactly $B$ clients to its facilities we have that $\sum_j x_{ij}=By_i$.
So each $i\leq n-1$ takes exactly $1-1/n^2$ demand from non-exclusive clients which  are
$(n-2)(B-1)$ in total. Thus, by symmetry of the construction, each one them is assigned to $i$
with a fraction of $\frac{B-1}{n^2(n-2)(B-1)}=\frac{1}{n^2(n-2)}$

\subsubsection{Proof of unbounded integrality gap of the constructed solution}

In the  present subsection, we  manipulate the costs of  instance $I$,
which we left undefined, so as to create a large integrality gap while
ensuring that the distances form a metric.

Set each facility opening cost to zero. As for the connection costs (distances)
consider the $(n-2)$-dimensional Euclidean space $\mathbb{R}^{n-2}$. Put
every facility $i,$  $i\leq n-1,$ together with its  exclusive clients on a
distinct vertex of an $(n-2)$-dimensional regular simplex with edge length
$D$. Put facilities $n,n+1$ together with their exclusive clients to a point
far away  from the simplex, so  the minimum distance from  a vertex is
$D' >> D.$ Setting $D'=\Omega(nD)$ is enough.

Since the distance between a facility and one of its exclusive clients
is  $0$,  the  cost  of  the fractional  solution  we  constructed  is
$O(nD)$. This cost  is due to the assignments  of exclusive clients of
facility $i,$ $i \leq n-1,$ to facilities $i'$ with $i' \neq i,$ $i' \leq n-1.$ 
As  for the cost  of an arbitrary integral  solution, observe
that since the $n^2+n-1$ exclusive  clients of $n,n+1$ are very far from
the  rest of  the facilities,  using $n$  of them  to  satisfy some
demand of  those facilities and help  to open all of  them, incurs a
cost of $\Omega(nD') = \Omega(n^2D).$ On the other hand, if we do not open all of the $n-1$
facilities on  the vertices of the  simplex (since they  have in total
$(n-1)(B-1)$  exclusive clients  which is  not enough  to open  all of
them), there  must be  at least  one such facility  not opened  in the
solution, thus its $B-1=\Theta(n^2)$ exclusive clients must be assigned elsewhere,
incurring a cost of  $\Omega(n^2D).$

This concludes the proof of Theorem~\ref{theorem:proper}.

\subsection{Proof of Theorem \ref{theorem:proper} for \cfl\ } 
\label{sec:proof_theorem_p2}

The proof is similar to that for \lbfl. 
We prove that the relaxation must use 
a specific set of classes and then we use these classes to construct a
desired feasible solution. In the last step we 
 define appropriately  the costs of the instance. 

\subsubsection{Existence of a specific type of classes}

Consider
an instance $I$ with $n$ facilities, where $n$ is
sufficiently large to ensure  that $\alpha n \leq n - c_0$  where
$c_0,$  is a
constant greater than or equal to  $1$. Let the capacity be $U=n^2$, and let
the number of  clients be $(n-1)U+1$. Notice that in every integer solution of the instance
 we must open  at least $n$ facilities. The  facility costs  and the assignment  costs will  be defined
later.  

 We  assume, like before, that  the  facilities are  numbered
$1,2,\ldots ,n$. 
Consider  an  integral  solution  $s$  for $I$   where  all the
facilities  are opened, and furthermore 
facilities $1,\ldots , n-1$  are assigned $U$ clients each 
and facility $n$ is assigned one client. 
Since our proper  relaxation is valid, there must be a solution $s'$ in the  space of
feasible solutions of the proper relaxation whose $(y,x)$ projection is the characteristic 
vector of $s$.  
By Definition~\ref{def:constell},
it is easy to see that $s'$ can only be obtained as a 
positive combination of classes $cl$ such that for
every    facility   $i$    we   have    $Clients_{cl}(i)   \subseteq
Clients_s(i)$.  Recall  that  since  the  complexity  of  our
relaxation is $\alpha$, the classes in the support of any solution 
have at most $n-c_0 \leq n-1$
facilities. 

Now consider the support  of $s'$. We will distinguish the classes $cl$ for
which variable $x_{cl}$ is in the support of $s'$ into 2 sets. The first set consists 
of the classes that assign exactly one client to facility $n$; call them \emph{type A} classes.
The second set  consists  of the classes that do not assign any client to facility 
$n$; call those \emph{type B} classes. By the discussion above those sets form a
partition of the classes in the support of $s'$, and moreover they are both non-empty: this is
 by the fact that at most  $n-c_0$ facilities are in any class, and by the fact
 that in $s$ all $n$ 
facilities are opened integrally. Notice also that no  class
of type B can contain facility $n$ even though the definition of a class does not
exclude the possibility that a class contains a facility to which no clients are
assigned. 

We call \emph{density} of  a class $cl$ the ratio 
$d(cl)=\frac{\sum_{i\neq n}|Clients_{cl}(i)|}{|F(cl)-\{n\}|}$. By the discussion 
above we have that $d(cl)\leq U$ for all $cl$ in the support of $s'$. The following holds:

\begin{lemma}
All classes in the support of $s'$ have density $U.$
\end{lemma}

\begin{proof}
The amount of demand that a class $cl$ contributes to the demand assigned to the set
of the first $n-1$ facilities by $s'$ is $d(cl)|F(cl)-\{n\}|x_{cl}.$
 We have $\sum_{cl}d(cl)|F(cl)-\{n\}|x_{cl}=(n-1)U$.
 Observe
that by the projection of $s'$ on $(y,x)$ and by the fact that for $i=1,\ldots ,n-1$,
 $y_i=1$ in $s$, we have $\sum_{cl}|F(cl)-\{n\}|x_{cl}=n-1$. 
Setting $m_{cl}=\frac{x_{cl}|F(cl)-\{n\}|}{n-1}$ we have from
the above $\sum_{cl}m_{cl}=1$ and $\sum_{cl}m_{cl}d(cl)=U$. 
The latter together with the fact that $d(cl)\leq U$ we have that $d(cl)=U$ for all classes
$cl$ in the support of $s'$.   
\end{proof}

The following corollary is immediate from the above:

\begin{corollary}
There is a type $B$ class in the support of $s'$ that has density $U.$
\end{corollary}

So far we have proved that 
in the class set of any proper relaxation for $I,$ there is 
a class $cl_0$ of type $B$ with density $d(cl_0)=C$.
 Let $|F(cl_0)| = t \leq n-1.$

\subsubsection{Construction of a bad solution}

Consider the symmetric classes of $cl_0$ for all permutations of the $n$ facilities
and for all permutations of the clients. Those classes are not necessarily in the support of $s'$. Take a quantity of measure $\epsilon$ and distribute it equally among all 
those classes. Since class $cl_0$ has density $U,$ all those symmetric classes
assign on average $U$ clients to each of their facilities. 
Due to symmetry, each facility is in a class $\epsilon \frac{t}{n}$ of the time and is assigned $\epsilon \frac{t}{n}U$ demand. Each client is assigned to
each facility $\epsilon\frac{tU}{((n-1)U+1)n}$ of the time. We call that step of our construction \emph{round $A$}.

Now consider the symmetric classes of $cl_0$ for all permutations of the first $n-1$ facilities
and for all permutations of the clients (those classes are well defined since $t\leq n-1$).
Again distribute a quantity of measure $\epsilon$ equally among all 
those classes. Similarly to the previous, each facility is in a class $\epsilon \frac{t}{n-1}$ of the time and is assigned $\epsilon \frac{t}{n-1}U$ demand. Each client is assigned to
each facility $\epsilon\frac{tU}{((n-1)U+1)(n-1)}$ of the time. 
We call that step of our construction \emph{round $B$}.

Spending $\phi=\frac{1}{nt}$ measure in round $A$ and $\xi=\frac{(n-1)(1-1/n^2)}{t}$ 
measure in round $B$ we construct a solution $s_b$ whose projection to $(y,x)$ is the 
following $(y^*,x^*)$:
$y^*_i=1$ for $i=1,\ldots ,n-1$, $y^*_n=\frac{1}{n^2}$, and for every client $j,$ $x^*_{nj}=\frac{U/n^2}{(n-1)U+1}$ and 
$x^*_{ij}=\frac{1-x^*_{nj}}{n-1}$ for $i=1,\ldots ,n-1.$ It is easy to see that $s_b$ is
a feasible solution for our proper relaxation.

Now simply set all distances to $0$, and define the facility opening costs as 
$f_n=1$ and $f_i=0$ for $i\leq n-1.$ It is easy to see
that the integrality gap of the proper relaxation is $\Omega (n^2)$.

\end{document}